\documentclass[a4paper,UKenglish,cleveref, autoref, thm-restate]{lipics-v2021}

\pdfoutput=1 
\hideLIPIcs  

\graphicspath{{./graphics/}} 

\title{Algorithm and Strategy Construction for Sure-Almost-Sure Stochastic Parity Games} 

\author{Laurent Doyen}
{Universit\'{e} Paris-Saclay, CNRS, ENS Paris-Saclay, LMF, Gif-sur-Yvette, France}{ldoyen@lmf.cnrs.fr}{https://orcid.org/0000-0003-3714-6145}{}
\author{Shibashis Guha}{Tata Institute of Fundamental Research, Mumbai, India}{shibashis@tifr.res.in}{https://orcid.org/0000-0002-9814-6651}{}
\authorrunning{L. Doyen and S. Guha}

\ccsdesc[100]{Mathematics of computing~Probability and statistics~Stochastic processes} 

\keywords{stochastic games, parity objectives, reactive synthesis} 

\category{} 

\relatedversion{} 

\nolinenumbers 

\usepackage[commandnameprefix=always]{changes}

\definechangesauthor[name={Laurent}, color=purple]{LD}
\definechangesauthor[name={Shibashis}, color=teal]{SG}

\usepackage{graphicx} 
\usepackage[pdflatex=true]{gastex}   

\usepackage{xcolor}
\presetkeys{todonotes}{disable}{}
\usepackage{tikz}
\usetikzlibrary{arrows.meta, positioning, matrix, shapes, automata,positioning,arrows,shapes.geometric}
\tikzset{
    ->,
    >=latex,
    node distance=1.3cm,
    every state/.style={thick, fill=gray!20, minimum size=3pt, inner sep=3pt},RequirePackage
    random/.style={diamond, thick, fill=gray!20, inner sep=2.5pt},
    square/.style={regular polygon,regular polygon sides=4, thick, fill=gray!20, inner sep=1.5pt},
    initial text=$ $,
}

\ifpdf
\else
  \AtBeginDocument{%
        \RemoveFromHook{shipout/firstpage}[hyperref]%
        \RemoveFromHook{shipout/before}[hyperref]%
  }
\fi

\RequirePackage{algorithm}
\RequirePackage{algpseudocodex}
\RequirePackage{listings}
\RequirePackage{xspace}

\RequirePackage{amsmath,amssymb,amsthm}

\algrenewcommand\algorithmicrequire{\textbf{Input:}}
\algrenewcommand\algorithmicensure{\textbf{Output:}}
\newcommand{\as}{{\sf AS}}
\newcommand{\sure}{{\sf S}}

\def\abs#1{\ensuremath{\lvert #1\rvert}} 
\def\closure#1{\ensuremath{\lceil #1\rceil}}
\newcommand{\Nat}{\mathbb{N}}

\let\epsilon\varepsilon
\let\emptyset\varnothing

\newcommand{\Prob}{\mathsf{Pr}}
\newcommand{\tuple}[1]{\langle #1 \rangle}
\newcommand{\dist}{{\mathcal D}}
\newcommand{\Supp}{{\sf Supp}}

\newcommand{\A}{\mathcal{A}}

\newcommand{\NP}{\mathrm{NP}}
\newcommand{\coNP}{\mathrm{coNP}}
\newcommand{\PTime}{\mathrm{Poly}}

\newcommand{\ASPG}{\mathsf{ASPG}}
\newcommand{\Parity}{\ensuremath{{\mathsf{Parity}}}}
\newcommand{\parity}{\ensuremath{p}}

\newcommand{\PO}{{Player~$1$}\xspace}
\newcommand{\PT}{{Player~$2$}\xspace}

\renewcommand{\Game}{\mathcal{G}} 

\newcommand{\Bnd}{{\sf Bnd}}

\newcommand{\EdgeValues}[2]{{\color{red}#1} \ifblank{#2}{}{ , {\color{blue}#2} }  }

\newcommand{\Outcome}{\mathsf{Out}}

\newcommand{\Last}{{\sf Last}}

\newcommand{\SurePre}{{\sf SurePre}}
\newcommand{\PosPre}{{\sf PosPre}}
\newcommand{\SureAttr}{{\sf SureAttr}}
\newcommand{\PosAttr}{{\sf PosAttr}}
\newcommand{\unlucky}{{\sf unlucky}}
\newcommand{\Unlucky}{{\sf Unlucky}}

\renewcommand\blue[1]{\textcolor{blue}{#1}} 
\definecolor{gray50}{gray}{0.45}
\def\gray50#1{\textcolor{gray50}{#1}}

\newcommand{\aut}{\mathcal{A}}
\newcommand{\proj}[2]{#1_{\upharpoonright #2}}

\newenvironment{longversion}{}{}

\newenvironment{shortversion}{}{}

\begin{document}
\excludecomment{shortversion}
\maketitle

\begin{abstract}
We consider turn-based stochastic two-player games with a combination 
of a parity condition that must hold surely, that is in all possible
outcomes, and of a parity condition that must hold almost-surely,
that is with probability~$1$. The problem of deciding the existence of
a winning strategy in such games is central in the framework of synthesis 
beyond worst-case where a hard requirement that must hold surely
is combined with a softer requirement. Recent works showed that
the problem is coNP-complete, and infinite-memory strategies are necessary in general,
even in one-player games (i.e., Markov decision processes). However, memoryless
strategies are sufficient for the opponent player.
Despite these comprehensive results, the known algorithmic solution 
enumerates all memoryless strategies of the opponent, which is exponential
in all cases, and does not construct a winning strategy when one exists. 

We present a recursive algorithm, based on a characterisation of the winning region, 
that gives a deeper insight into the problem. In particular, we show how to 
construct a winning strategy to achieve the combination of sure and almost-sure
parity, and we derive new complexity and memory bounds for special classes
of the problem, defined by fixing the index of either of the two parity conditions.
\end{abstract}

\section{Introduction}

Turn-based games provide a simple model for the synthesis of system components 
satisfying a given specification. A component is viewed as a player
trying to achieve a winning condition defined by the specification, 
regardless of the behaviour of other components. Abstracting the other
components as a combination of an adversarial and stochastic environment
gives rise to the standard model of two-player stochastic games played on graphs~\cite{FV97,automata},
where the vertices of the graph are partitioned into system, environment, and stochastic vertices. 
Starting from an initial vertex, a play is an infinite path in which the successor 
of each vertex is chosen by the player owning that vertex (for system and environment vertices)
or according to a probability distribution (for stochastic vertices). In general,
strategies define the choice of the players, which may depend on the sequence of 
previously visited vertices, called the history. The synthesis of a component
corresponds to constructing a strategy for the system that maximizes the probability
to win (i.e., to satisfy the winning condition), regardless of the strategy of the environment.

We consider specifications described by parity conditions, which are  
canonical to express all $\omega$-regular requirements on plays~\cite{Thomas97,CN06}, 
with the property that memoryless optimal strategies exists for both 
players~\cite[Theorem~1.2]{CJH04}, maximizing the probability to win without having to remember
the history of the game. Deciding which player  
has a winning strategy in parity games lies in NP~$\cap$~coNP~\cite[Corollary~1.1]{CJH04}, 
but no polynomial-time algorithm is known for that problem, even in non-stochastic games.  
Improved bounds and quasi-polynomial algorithms exist though~\cite{Jur98,CJKLS17,JL17,Leh18}.

In component synthesis, a natural requirement is to ensure that the winning condition
holds with probability~$1$, called \emph{almost-sure winning}. For example, to 
send a message over a network that may lose messages with a fixed probability $p$,
a strategy is to keep sending the same message until delivery, 
which is almost-sure winning because the probability of delivery failure 
after $k$ attempts is vanishing, as $p^k \to 0$ when $k \to \infty$.
An even stronger requirement is that all possible outcomes, no matter their probability 
should, satisfy the specification, called \emph{sure winning}. This corresponds
to the worst-case scenario where the choices at both environment and 
stochastic vertices are purely antagonistic. 
Sure winning is typically used for critical requirements, while almost-sure winning
is used to ensure good performance in a stochastic environment.
Specifications that combine both a sure and almost-sure parity condition
have been considered recently~\cite{BRR17,CP19}, as a fundamental instance of 
the framework of synthesis beyond worst-case~\cite{BFRR17}. 
For such specifications, games are determined~\cite[Theorem~5]{CP19} and 
the synthesis problem is coNP-complete~\cite[Corollary~7]{CP19};
memoryless strategies are sufficient for the environment player~\cite[Theorem~6]{CP19}, and infinite
memory is necessary for the system player~\cite[Example~1]{BRR17}, even in MDPs (Markov decision processes,
i.e., games with no environment vertices).

Despite these comprehensive results, the picture is not complete. The first 
issue is algorithmic. The coNP algorithm~\cite{CP19} is to guess a memoryless strategy
for the environment player, and then to check that the system player is losing
in the resulting MDP (using determinacy), which can be done in NP~$\cap$~coNP~\cite[Theorem~15]{BRR17}.
In practice, enumerating the memoryless strategies for the environment is always exponential,
thus inefficient.
Another related issue is the construction of a winning strategy for the system
player, which is central in the context of synthesis. Only in the case of MDPs,
the structure of (infinite-memory) winning strategies has been described~\cite[Definition~10]{BRR17}.
However, the current solution of the synthesis problem for games gives no clue 
about how the system should play in order to satisfy the specification.
We argue that these issues reflect the lack of a deeper insight into games 
with a combination of sure and almost-sure parity conditions. 

\subparagraph*{Contribution}
We outline the main contributions of the paper.
\begin{itemize}
\item We present a recursive algorithm for deciding the existence of a strategy
that is both sure winning for a first parity condition, and almost-sure winning
for a second parity condition. The algorithm computes the winning set of
the system player in worst-case exponential time. The correctness is established
using characterisations of the winning sets for each player, and gives a better
insight into the structure and properties of such games. 

\item We obtain a precise description of the winning strategies
for the system player, and how to construct them by switching between simple (memoryless) 
strategies, and repeat them for long enough using unbounded counters. 
Intuitively, a strategy that is only almost-sure winning for both parity conditions
may have violating outcomes, but representing a mass of probability~$0$.
Eventually switching to a sure-winning strategy for the first parity condition is necessary to ensure all outcomes satisfy that condition, however possibly at the price of violating the second condition.
The crux is to switch so rarely that those violating outcomes represent a mass of probability $0$, maintaining the almost-sure satisfaction of the second parity condition.
A similar idea works in MDPs, combining strategies that almost-surely visit 
certain end-components~\cite{BRR17}. However, the extension to stochastic games 
is non-trivial as the notion of end-components no longer exists in games.

\item As a consequence of the new algorithm and analysis, we get new bounds
on the complexity and memory requirement of subclasses of the problem,
involving the parity index (i.e., the number of priorities defining a parity condition). 
When either of the two parity conditions has a fixed index, we show that 
the synthesis problem becomes solvable in NP~$\cap$~coNP; when both have 
a fixed index, the problem is in P. When the sure condition is coB\"uchi,
memoryless strategies are sufficient for the system player,
and when the almost-sure condition is B\"uchi,
finite memory is sufficient. 
This gives a tight and complete picture, as it is known
that the combination of a sure B\"uchi and almost-sure coB\"uchi condition
requires infinite memory, already in MDPs~\cite[Example~1 \& Theorem~18]{BRR17}.

\item A technical result of independent interest is the construction of a deterministic parity automaton (DPW)
equivalent to the intersection of two DPW, with a blow up that is exponential 
only in the parity index of either of the two DPW. Beyond the product of the state
spaces of the two automata, the blow up is by a factor $\sqrt{\min(d_1^{d_2},d_2^{d_1})}$
where $d_i$ ($i=1,2$) is the parity index of the $i$-th automaton. 
Our solution is not particularly involved, but we could not find a construction with this property in the literature.
The construction is crucial in proving that the synthesis problem is in $\NP \cap \coNP$ when one of the parity conditions has a fixed index.
\end{itemize}




\subparagraph*{Related work}
Synthesis of strategies that satisfy a combination of multiple objectives 
that arises naturally in many applications, such as verification~\cite{AKV16}, 
planning~\cite{GZ18}, and AI~\cite{RVWD13}.
In non-stochastic games and thus for sure winning, combinations of parity objectives~\cite{CHP07}
and of mean-payoff objectives~\cite{VCDHRR15} have been considered.
In stochastic games, the quantitative analysis of multi-objective conditions saw very little 
progress, in particular for Boolean objectives (such as the parity condition), and the known 
positive results hold for reachability objectives or subclasses of games~\cite{CFKSW13,BKW24}.
As observed in~\cite{BRR17}, combinations of quantitative objectives is sometimes conceptually
easier than Boolean objectives. Almost-sure winning for multiple quantitative objectives
has been studied in the case of mean-payoff conditions~\cite{CD16a}. 
The framework of beyond-worst-case synthesis was initially studied
with quantitative objectives such as shortest path and mean-payoff~\cite{BFRR17}.

We already discussed the results that are closest to our work~\cite{BRR17,CP19}. 
We note that the notion of limit-sure winning (i.e., winning with probability arbitrarily close to~$1$) has also been considered in
combination with sure winning, and the synthesis problem is coNP-complete,
where the coNP bound is obtained by guessing a memoryless strategy for
the environment player, and then solving an MDP with sure-limit-sure parity condition,
which is shown to be polynomial-time reducible to the sure-almost-sure problem
in MDPs~\cite{CP19}. Boolean combinations of sure and almost-sure winning
have been considered for MDPs~\cite{BGR20}. A natural extension would be to consider arbitrary probability
of satisfaction (instead of probability~$1$) in combination with sure winning. 
This question is solved for MDPs~\cite{BRR17}, and remains open for games.

\begin{shortversion}
Omitted proofs are available in the extended version in the appendix.
\end{shortversion}





\section{Preliminaries}

\subsection{Basic definitions}
A \emph{probability distribution} on a finite set~$V$ is a
function $d : V \to [0, 1]$ such that $\sum_{v \in V} d(v)= 1$. 
The \emph{support} of~$d$ is the set $\Supp(d) = \{v \in V \mid d(v) > 0\}$. 
We denote by $\dist(V)$ the set of all probability distributions on~$V$. 
Given $v \in V$, we denote by $1_v$ the \emph{Dirac distribution} on~$v$ that 
assigns probability~$1$ to~$v$.

A \emph{stochastic game} $\Game = \tuple{V, (V_1,V_2,V_{\Diamond}), E, \delta}$ consists 
of a finite set $V$ of vertices, 
mutually disjoint (but possibly empty) sets $V_1,V_2,V_{\Diamond}$ 
of respectively Player $1$, Player $2$, and probabilistic vertices 
such that $V =  V_1 \cup V_2 \cup V_{\Diamond}$,
a set $E \subseteq V \times V$ of edges, 
and a probabilistic transition function $\delta: V_{\Diamond} \to \dist(V)$.
We denote by $E(v) = \{v' \mid (v,v') \in E \}$ the set of all successors of a vertex $v \in V$,
and we require that every vertex has at least one successor, $E(v) \neq \emptyset$ for all vertices $v \in V$,
and that $\Supp(\delta(v)) = E(v)$ for all probabilistic vertices $v \in V_{\Diamond}$.
The decision problems in this paper are independent of the transition probabilities, and therefore  stochastic games can be specified by giving only the partition $(V_1,V_2,V_{\Diamond})$
and the set $E$ of edges.
A \emph{Markov decision process (MDP)} for Player~$i$ is the special case
of a stochastic game where one of the players has no role, that is $V_{3-i} = \emptyset$.

The game is played in rounds, starting from an initial vertex $v_0$. 
In a round at vertex $v$, if $v \in V_i$ ($i=1,2$), then Player~$i$ chooses a successor $v' \in E(v)$ (which is always possible since $E(v) \neq \emptyset$),
and otherwise $v \in V_{\Diamond}$ and a successor $v' \in E(v)$ is chosen with probability $\delta(v)$.
The next round starts at vertex $v'$.
A \emph{play} is an infinite path $v_0, v_1, \dots$ such that $(v_j,v_{j+1}) \in E$ for all $j \geq 0$.
A \emph{history} is a finite prefix of a play. For $\tau = v_0, v_1, \dots, v_k$ and $j \leq k$,
let $\abs{\tau} = k+1$ be the length of $\tau$, define $\Last(\tau) = v_k$, and $\tau(j \rhd) = v_j, \dots, v_k$.
Given a function $\Omega: V \to Z$ defined on vertices, define $\Omega(\tau) = \Omega(v_0), \Omega(v_1), \dots, \Omega(v_k)$
the extension of $\Omega$ to histories.

A \emph{strategy} for Player~$i$ ($i=1,2$) in $\Game$ is a function $\sigma_i: V^* V_i \to V$ 
such that $(v,\sigma_i(\tau v)) \in E$ is an edge for all histories $\tau \in V^*$ and vertices $v \in V_i$.
A play $v_0, v_1, \dots$ is an \emph{outcome} of $\sigma_i$ if $v_{j+1} = \sigma_i(v_0, v_1, \dots, v_j)$
for all $j \geq 0$ such that $v_j \in V_i$. We denote by $\Outcome(v,\sigma)$
the set of all outcomes of a strategy $\sigma$ with initial vertex $v$,
and by $\Outcome(v,\sigma_1,\sigma_2)$ the set $\Outcome(v,\sigma_1) \cap \Outcome(v,\sigma_2)$.    
We denote by $\Prob_{v_0}^{\sigma_1,\sigma_2}$ (or simply $\Prob_{v_0}^{\sigma_i}$ in the case of MDPs for Player~$i$)
the standard probability measure 
on the sigma-algebra over the set of (infinite) plays with initial vertex $v_0$, 
generated by the cylinder sets spanned by the  histories~\cite{BK08}.

A strategy $\sigma_i$ uses \emph{finite memory} if there exists a right congruence $\approx$
of finite index (i.e., that can be generated by a finite-state transducer) over 
histories such that $\tau \approx \tau'$ implies $\sigma_i(\tau)  = \sigma_i(\tau')$.

An objective is a measurable set $\psi \subseteq V^{\omega}$ of plays.
Given a priority function $\Omega:V \to \Nat$, we consider the classical parity objective
$\Parity(\Omega) = \{v_0, v_1, \dots \in V^{\omega} \mid \limsup_{j \to \infty} \Omega(v_j) \text{ is even} \}$
and denote by $\Parity^\textsf{c}(\Omega) = V^{\omega} \setminus \Parity(\Omega)$
the complement of that objective (which can itself be defined as a parity objective, using
priority function $\Omega^{+1}$ where $\Omega^{+1}(v) = \Omega(v) + 1$).
B\"uchi objectives are the special case of priority functions $\Omega:V \to \{1,2\}$ that require visiting priority~$2$ infinitely often, while coB\"uchi objectives are the special case of priority functions $\Omega:V \to \{0,1\}$ that require visiting priority $1$ only finitely often.
A key property of parity objectives is prefix-independence: $\theta \in \Parity(\Omega)$
if and only if $\tau \theta \in \Parity(\Omega)$, for all $\theta \in V^{\omega}$ and $\tau \in V^*$.
Given a set $U \subseteq V$ and priority $p \in \Nat$, we denote by $U(\Omega, p) = \{v \in U \mid \Omega(v) = p \}$
the set of vertices of $U$ with priority $p$, and by $\pi_{\Omega \geq p}(\tau)$
the sequence obtained from $\Omega(\tau)$ by removing all priorities smaller than $p$. 
For example, if $\tau = v_2 v_3 v_1 v_3 v_4 v_1 v_3$ and $\Omega(v_i) = i$,
then $\pi_{\Omega \geq 3}(\tau) = 3343$.

A strategy $\sigma_1$ of Player~$1$ is \emph{sure winning} for objective $\psi$ from an initial vertex $v$ if $\Outcome(v,\sigma_1) \subseteq \psi$,
and \emph{almost-sure winning} if $\Prob_{v}^{\sigma_1,\sigma_2}(\psi) = 1$
for all strategies $\sigma_2$ of Player~$2$.
Note that sure winning for $\psi$ immediately implies almost-sure winning for $\psi$ (not the converse),
and that a strategy that is almost-sure winning for $\psi_1$ and for $\psi_2$
is almost-sure winning for $\psi_1 \cap \psi_2$ (here the converse holds).

We consider the sure-almost-sure synthesis problem for parity objectives,
which is to decide, given a stochastic game $\Game$ with initial vertex $v$
and two priority functions $\Omega_1,\Omega_2$, whether there exists a strategy $\sigma_1$ of Player~$1$
that is both sure winning for objective $\Parity(\Omega_1)$ and almost-sure winning for objective $\Parity(\Omega_2)$.
We call such a strategy \emph{sure-almost-sure} winning (assuming that the game $\Game$,
initial vertex $v$, and priority functions $\Omega_1,\Omega_2$ are clear from the context).
We call the set of all initial vertices 
for which such a strategy exists
the sure-almost-sure winning region, and we call the complement the
sure-almost-sure spoiling region. 
Sure winning regions and almost-sure winning regions are defined analogously.

It is known that finite-memory strategies are not sufficient for sure-almost-sure winning,
already in MDPs with a combined sure B\"uchi and almost-sure coB\"uchi objective~\cite{BRR17},
as illustrated in \figurename~\ref{fig:infinite-memory} (in figures, Player-$1$ vertices are shown as circles,
Player-$2$ vertices as boxes, and probabilistic vertices as diamonds). 
There is a strategic choice only in $v_a$. 
Surely visiting infinitely often a vertex in $\{v_c,v_d\}$, and at the same time almost-surely visiting only finitely often $v_d$ can be done by choosing $v_b$ from $v_a$ more and more often as compared to choosing $v_d$, in a way such that choosing $v_d$ infinitely often has probability~$0$, yet $v_d$ is visited infinitely often in all outcomes where $v_c$ is no longer visited from some point on.

\begin{figure}[!t]
   \begin{center}
	\input{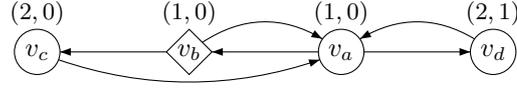}
   \end{center}
\caption{An MDP where the combination of sure B\"uchi and almost-sure coB\"uchi requires infinite memory~\cite{BRR17}. 
Each node $v$ is labeled by the pair of priorities $(\Omega_1(v), \Omega_2(v))$,
inducing a (sure) B\"uchi condition to visit a vertex in $\{v_c,v_d\}$ infinitely often, and an (almost-sure) coB\"uchi condition to visit $v_d$ finitely often.
 \label{fig:infinite-memory}}
\end{figure}



\todo{define $\sure(\psi)$ and $\as(\psi)$ ? (if used)}

\subsection{Subgames and traps} \label{sec:subgames_and_traps}
Let $\SurePre_i: 2^V \to 2^V$ and $\PosPre_i: 2^V \to 2^V$ 
be respectively the sure and positive predecessor operator for Player~$i$ defined, 
for all $U \subseteq V$, 
by $\SurePre_i(U) = \{v \in V_i \mid E(v) \cap U \neq \emptyset \} 
\cup \{v \in V_{3-i} \cup V_{\Diamond} \mid E(v) \subseteq U \}$,
and by $\PosPre_i(U) = \{v \in V_i  \cup V_{\Diamond} \mid E(v) \cap U \neq \emptyset \} 
\cup \{v \in V_{3-i} \mid E(v) \subseteq U \}$.

Given a set $T \subseteq V$, the sure attractor $\SureAttr_i(T)$ for Player~$i$ 
is the least fixed point of the operator $X \mapsto \SurePre_i(X) \cup T$,
thus $\SureAttr_i(T) = \bigcup_{j \geq 0} \SurePre_i^j(T)$ (where $\SurePre_i^0(T) = T$)
and the positive attractor $\PosAttr_i(T)$ for Player~$i$ is the least fixed
point of the operator $X \mapsto \PosPre_i(X) \cup T$.

Intuitively, from the vertices in $\SureAttr_i(T)$ Player~$i$ has a memoryless 
strategy to ensure eventually (in fact, within at most $\abs{V}$ rounds) reaching 
a vertex in $T$ regardless of the choices of Player~$3-i$ 
and of the outcome of the probabilistic transitions; from the vertices in $\PosAttr_i(T)$ 
Player~$i$ has a memoryless 
strategy to ensure reaching with positive probability $T$ regardless of the choices of Player~$3-i$.
We call such strategies the attractor strategies.

A set $U \subseteq V$ induces a \emph{subgame} of $\Game$, denoted by $\Game_{\upharpoonright U}$,
if $E(v) \cap U \neq \emptyset$ for all $v \in U \cap (V_1 \cup V_2)$, and $E(v) \subseteq U$ 
for all $v \in U \cap V_{\Diamond}$, that is, the two players can cooperate to keep
the play forever in $U$. Such a set $U$ is a \emph{trap} for Player~$i$ if, moreover, 
$E(v) \subseteq U$ for all $v \in U \cap V_{i}$. For example, the set $V \setminus \PosAttr_i(T)$
is always a trap for Player~$i$. 
A key property of traps is that if Player~$i$ has a sure (resp., almost-sure) winning strategy in $\Game_{\upharpoonright U}$
(for some objective $\psi$)
from an initial vertex $v \in U$ where $U$ is a trap for Player~$3-i$,
then Player~$i$ has a sure (resp., almost-sure) winning strategy in $\Game$ from $v$ for $\psi$. 

Consider a set $U \subseteq V$ that only satisfies the condition $E(v) \cap U \neq \emptyset$ for all $v \in U \cap (V_1 \cup V_2)$,
as it is the case for $U = V \setminus \SureAttr_i(T)$.
Let $\Bnd(U) = \{ v \in U \cap V_{\Diamond} \mid E(v) \not \subseteq U\}$ be the set 
of (boundary) probabilistic vertices in $U$ with a successor outside $U$. 
Define the \emph{subgame closure} $\closure{\Game}_{U} = \tuple{U', (U_1,U_2,U_{\Diamond}), E', \delta'}$ 
by redirecting the (probabilistic) transitions leaving $U$ to a sink vertex, formally (see also \figurename~\ref{fig:even}):
\begin{itemize}
\item $U' = U \cup \{v_{\sf sink}\}$, 
\item $U_1 =  U \cap V_1$; $U_2 =  U \cap V_2$; $U_{\Diamond} =  (U \cap V_{\Diamond}) \cup \{v_{\sf sink}\}$,
\item $E' = E \cap (U \times U) \ \cup\  \Bnd(U) \times \{ v_{\sf sink} \} \ \cup\  \{ (v_{\sf sink},v_{\sf sink})\}$,
\item $\delta'(v)(v') = \delta(v)(v')$ for all $v \in U_{\Diamond}$ and $v' \in U$; $\delta'(v)(v_{\sf sink}) = \sum_{v' \in V \setminus U} \delta(v)(v')$ for all $v \in U_{\Diamond}$; and $\delta'(v_{\sf sink}) = 1_{v_{\sf sink}}$.
\end{itemize}
Note that if $U$ induces a subgame of $\Game$, then $\closure{\Game}_{U}$ is the disjoint union
of $\Game_{\upharpoonright U}$ and $\{v_{\sf sink}\}$.
We generally assign priority $\Omega_1(v_{\sf sink}) = \Omega_2(v_{\sf sink}) = 0$ to the sink vertex,
making it winning for Player~$1$. Then the key property is that 
if Player~$1$ does not have a sure (resp., almost-sure) winning strategy in $\closure{\Game}_{U}$
(for some objective $\psi$) from an initial vertex $v \in U$,
then Player~$1$ does not have a sure (resp., almost-sure) winning strategy in $\Game$ from $v$ for $\psi$\todo{Is this the case only when the maximum priority $d$ is odd?}.

\section{Conjunction of Parity Conditions}\label{sec:conjParity}

It will be useful in the sequel to consider conjunctions of parity conditions: Our algorithm (Section~\ref{sec:alg}) needs to compute the almost-sure winning region for the conjunction $\Parity(\Omega_1) \cap \Parity(\Omega_2)$ of two parity objectives.

We revisit the problem of transforming an objective specified as the conjunction
of two parity conditions into an objective specified by a single parity condition.
The transformation is always possible because the parity condition is a canonical 
way to express $\omega$-regular objectives, and the class of $\omega$-regular objectives
is closed under Boolean operations~\cite{Thomas97}. However, we were unable to
find a direct construction in the literature, in particular a construction
that would incur only a polynomial blowup when either of the two priority
functions has a fixed range (i.e., when the number of priorities is fixed 
in one of the parity conditions).

Such a construction being of independent interest (e.g., in formal languages
and automata theory), we present it using automata models without referring to
objectives and games. A \emph{deterministic automaton} over a finite alphabet 
$\Sigma$ is a tuple $\A=\tuple{Q,q_{\epsilon},\delta,\alpha}$ consisting 
of a finite state space $Q$, an initial state $q_{\epsilon} \in Q$,
a transition function $\delta: Q \times \Sigma \to Q$, and an acceptance
condition (which defines a subset of $Q^{\omega}$), here a conjunction of parity conditions.
A run of $\A$ on an infinite word $w = \sigma_0 \sigma_1 \dots$ is a 
sequence $q_0 q_1 \dots$ such that $q_0 = q_{\epsilon}$ and $\delta(q_i,\sigma_i) = q_{i+1}$
for all $i\geq 0$, hence every word has a unique run. 
Given priority functions $\Omega_1,\Omega_2: Q \to \Nat$, the set $L(\A)$ of all words
whose run satisfies the parity conditions $\Parity(\Omega_1)$ and $\Parity(\Omega_2)$
is called the \emph{language} of $\A$. We denote such automata by D2PW,
and in the special case where $\Omega_1 = \Omega_2$ by DPW, which is the classical
deterministic parity automata.

The blowup in the transformation of D2PW into DPW is inevitably exponential,
as shown in~\cite[Theorem~9]{Boker18} and confirmed by the fact that games
with a conjunction of two parity objectives are coNP-complete~\cite{CHP07},
whereas with a single parity objective they are in NP~$\cap$~coNP~\cite{EJ91}.  
A natural path for this transformation is to view the parity condition 
as a special case of Streett conditions~\cite{Boker18} which are closed under intersection,
then to convert the Streett automaton (obtained from the D2PW) into a DPW.
Priority functions with range $[d_1] = \{0,\dots,d_1\}$ and $[d_2] = \{0,\dots,d_2\}$
give $k = O(d_1 + d_2)$ Streett pairs, and the blowup in the conversion to DPW 
is by a factor in $O(k^k)$~\cite{Safra92}. A key observation is that this is an exponential blowup,
even if one of $d_1$ or $d_2$ is constant. The approach ``flattens'' the parity
conditions by ``merging'' them, which also happens when using (nondeterministic) B\"uchi automata\footnote{There is a quadratic translation of DPW into nondeterministic B\"uchi automata,
the intersection of B\"uchi automata gives a B\"uchi automaton with polynomial blowup,
but going back to DPW is exponential~\cite{Piterman07}. This approach also
``flattens'' the parity condition in the sense that the exponential is in both $d_1$ 
and $d_2$.}~\cite[Section~3.1]{Boker18}.

In contrast, we present a direct construction of DPW from D2PW that is conceptually simple and with a blowup of at most
$O(\max(d_1^{d_2/2},d_2^{d_1/2}))$, thus polynomial if either of the priority functions has a fixed range.

\paragraph*{The construction}

\begin{figure}[!t]
   \begin{center}
	\input{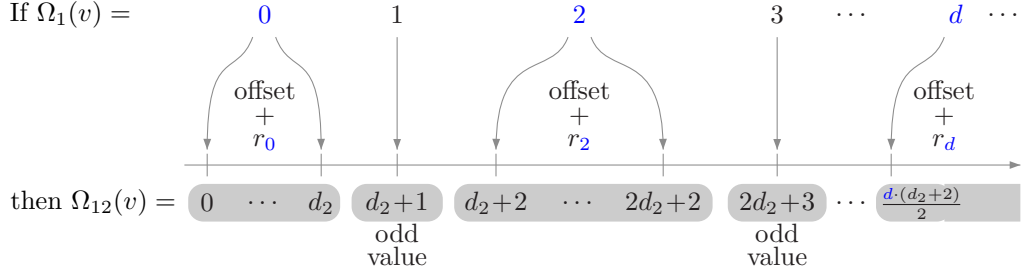}
   \end{center}
\caption{Construction of the priority function $\Omega_{12}$ for the conjunction of
the parity condition defined by the priority functions $\Omega_{1}$ and $\Omega_{2}$.
The offset for an even $\Omega_{1}$-priority \blue{$d$} is calculated as $\frac{\blue{d}}{2}\cdot (d_2 + 2)$.
The register $r_{\blue{d}}$ stores the largest $\Omega_{2}$-priority since the last
visit to an $\Omega_{1}$-priority \blue{$d$}.
 \label{fig:map}}
\end{figure}

Given a D2PW where the largest priority in the two parity conditions are $d_1$
and $d_2$, 
we construct an equivalent DPW where the largest priority is $O(d_1\cdot d_2)$.
We informally refer to priorities defined by a priority function $\Omega_{i}$ 
as $\Omega_{i}$-priorities. It will be convenient to assume that $d_2$ is even, 
which is without loss of generality (we may increment $d_2$ if $d_2$ is odd).

Intuitively, given a finite run $\tau = q_0 q_1 \dots q_k$ on some finite word $w$ in the D2PW,
where the priorities are $\Omega_1(\tau) = a_0 a_1 \dots a_k$ and $\Omega_2(\tau) = b_0 b_1 \dots b_k$, 
the run of our DPW on $w$ has priority sequence $c_0 \dots c_k$ with the following properties 
(see also \figurename~\ref{fig:map}): 

\begin{itemize}
\addtolength{\itemindent}{2mm}
\item[($P_1$)] \hspace{-2mm}if the $\Omega_1$-priority $a_k$ is odd, then $c_k$ is odd
and the value of $c_k$ depends only on $a_k$;

\item[($P_2$)]  \hspace{-2mm}for all runs $\tau'$ on some finite word $w'$ in the D2PW,
with $\Omega_1(\tau') = a'_0 \dots a'_k$, if $a_k < a'_k$ then
$c_k < c'_k$ (where $c'_k$ is the priority of the last state in the run of our DPW on $w'$); 
this is achieved by allocating disjoint segments where the value of $c_k$ may lie. 
The segment associated with $a_k$ has range $[d_2]$ if $a_k$ is even, 
and it is a single odd value if $a_k$ is odd;

\item[($P_3$)]  \hspace{-2mm}if the $\Omega_1$-priority $a_k$ is even, then the value of $c_k$
is calculated as the sum of an offset~$\nu(a_k)$ corresponding to the origin of the segment
associated with  $a_k$ and the largest $\Omega_{2}$-priority since the previous 
visit to the $\Omega_{1}$-priority $a_k$ (or since the beginning of the run
if $a_k$ is visited for the first time in $\tau$). That is, $c_k = \nu(a_k) + \max \{b_i, b_{i+1}, \dots, b_{k-1}\}$
where $i$ is the largest index smaller than $k$ such that $a_i = a_k$ (or, $i=0$).
\end{itemize}

In order to satisfy ($P_3$), the DPW automaton maintains, for each 
even $\Omega_1$-priority $d$,  a register $r_d$ that contains 
the largest $\Omega_{2}$-priority since the previous visit to 
a state with $\Omega_1$-priority $d$ (or since the beginning of the run).

Given a D2PW $\A=\tuple{Q,q_{\epsilon},\delta,\alpha}$ where the
acceptance condition $\alpha$ is a conjunction of parity conditions
defined by two priority functions $\Omega_{1}$ and $\Omega_{2}$,
let $\A'=\tuple{Q',q'_{\epsilon},\delta',\alpha'}$ be the DPW defined
as follows:
\begin{itemize}
\item $Q' = Q \times [d_2]^{[d_1]_{\mathrm{even}}}$ where $[d_1]_{\mathrm{even}} = \{0,2,\dots,2\lfloor\frac{d_1}{2}\rfloor\}$; a state $(q,r_0,r_2,\dots, r_{2\lfloor\frac{d_1}{2}\rfloor}) \in Q'$
consists of a state $q$ of $\A$ and, for each even $\Omega_1$-priority $d$,
a register $r_d$ 
with range $[d_2] = \{0,\dots,d_2\}$; 
\item $q'_{\epsilon} = (q_{\epsilon},0,0,\dots, 0)$;
\item for every state $(q,r_0,r_2,\dots, r_{2\lfloor\frac{d_1}{2}\rfloor}) \in Q'$ and $\sigma \in \Sigma$, define $\delta( (q,r_0,r_2,\dots, r_{2\lfloor\frac{d_1}{2}\rfloor}), \sigma) = (q',r'_0,r'_2,\dots, r'_{2\lfloor\frac{d_1}{2}\rfloor})$
where $q' = \delta(q,\sigma)$ and for every $d \in [d_1]_{\mathrm{even}}$:
$$r'_d =   \begin{cases}
    \max(r_d,\Omega_2(q))  & \text{if } \Omega_1(q) \neq d \\
    \Omega_2(q)  & \text{if } \Omega_1(q) = d 
  \end{cases}$$
\item the acceptance condition $\alpha'$ is a parity condition defined by the 
priority function $\Omega_{12}$ 

such that $\Omega_{12}(q,r_0,r_2,\dots, r_{2\lfloor\frac{d_1}{2}\rfloor}) = \begin{cases} 
    \frac{d\cdot(d_2 + 2)+d_2}{2}   & \text{if } d = \Omega_1(q) \text{ is odd} \\
    \frac{d\cdot(d_2 + 2)}{2} + r_d & \text{if } d = \Omega_1(q) \text{ is even}
  \end{cases}$

The value $\nu(d) = \frac{d\cdot(d_2 + 2)}{2}$ is called the \emph{offset} corresponding
to even $\Omega_1$-priority $d$ (see \figurename~\ref{fig:map}).
\end{itemize}

\begin{shortversion}
An example illustrating the construction is given in the extended version.    
\end{shortversion}

\begin{longversion}
\begin{example}
In Figure~\ref{fig:conj-parities}, on the left side, we have an input D2PW $\aut$ with states $Q$ and alphabet $\{a,b\}$.
To accept a word, it must visit $q_3$ infinitely often.
In the figure, above every state $q \in Q$, we write the priorities $\Omega_1(q)$ and $\Omega_2(q)$ in parentheses.
Thus $\Omega_1(q_0) = \Omega_1(q_1) = 0$, $\Omega_1(q_2) = 2$, $\Omega_1(q_3) = 1$ and $\Omega_2(q_0) = 5$, $\Omega_2(q_1) = 3$, $\Omega_2(q_2) = 1$, $\Omega_2(q_3) = 6$.

In the figure on the right, we show the constructed automaton $\aut'$ with states $Q'$.
Note that the only even priorities that appear on the first dimension for automaton $\aut$ are $0$ and $2$.
As stated above, the initial state of $\aut'$ is $q_0' = (q_0, 0, 0)$.
As the automaton $\aut$ reads a letter and moves to state $q_1$ from $q_0$, then the maximum priority seen on the second dimension since the last visit to a state $q$ with $\Omega_1(q) = 0$ is 5 and the maximum priority seen on the second dimension since the last visit to a state $q$ with $\Omega_1(q) = 2$ is also 5.
Hence $\aut'$ moves to a state $(q_1,5,5)$ from the initial state $(q_0,0,0)$.
The automaton $\aut$ then reads another letter to move to $q_2$ from $q_1$.
For $d=0$, we have that $q_1$ is the last state $q$ visited before the current visit to $q_2$ such that $\Omega_1(q) = 0$.
Since $\Omega_2(q_1) = 3$, we now move to the state $(q_2,3,5)$ in $\aut'$.
From $q_2$, if $\aut$ now reads a letter and moves to $q_0$, since $\Omega_1(q_2)=2$ and $\Omega_2(q_2)=1$, we move to the state $(q_0,3,1)$ in $\aut'$.
The other states in $\aut'$ can be defined similarly.

Now we describe the priority function $\Omega_{12}$.
Note that $d_2=6$.
Since $\Omega_1(q_0)=0$ which is even, we have that $\Omega_{12}(q_0,0,0) = \frac{0 \cdot 8}{2} + r_0 = 0+0=0$.
Again, since $\Omega_1(q_1)=0$, we have that $\Omega_{12}(q_1,5,5)=\frac{0 \cdot 8}{2} + r_0 = 0+5=5$.
For the state $(q_2,3,5)$, since $\Omega_1(q_2)=2$, we have that $\Omega_{12}(q_2,3,5)=\frac{2 \cdot 8}{2} + r_2 = 8+5=13$.
For the state $(q_3,3,1)$, since $\Omega_1(q_3)=1$ which is odd, we have that we have that $\Omega_{12}(q_3,3,1)=\frac{1 \cdot 8 + 6}{2} = 7$.
For the other states in $Q'$, the parity function $\Omega_{12}$ can be defined analogously.
The priority corresponding to each state in $\aut'$ appears above the state in the figure.
Further, for every transition from a state $(q, r_2, r_4)$ to a state $(q', r'_2, r'_4)$ in $\aut'$, we have a corresponding transition from state $q$ to $q'$ in $\aut$ and the letter labelling the transition in $\aut'$ is the same as the letter labelling the transition in $\aut$.
To avoid clutter, we did not add these letters on the transitions in the figure describing $\aut'$.
\end{example}

\begin{figure}
\begin{minipage}{.3\textwidth}
\begin{tikzpicture}[->, >=Stealth, node distance=2cm]

  \node[circle, draw, minimum size=6mm] (A) at (0,0) {$q_0$};
  \node[above=0mm of A] {(0,5)};
  \node[circle, draw, minimum size=6mm] (B) at (2,0) {$q_1$};
  \node[above=0mm of B] {(0,3)};
  \node[circle, draw, minimum size=6mm] (C) at (0,-2) {$q_2$};
  \node[above left=-1mm of C] {(2,1)};
  \node[circle, draw, minimum size=6mm] (D) at (2,-2) {$q_3$};
  \node[above=0mm of D] {(1,6)};
  \node at (-1,0) (Invis) {}; 

  \draw[->] (Invis) -- (A);           
  \draw[->] (A) -- (B) node[midway, above] {$a,b$};               
  \draw[->] (B) -- (C) node[midway, above] {$a,b$};;              
  \draw[->] (C) -- (A) node[midway, right] {$a$}; 
  \draw[->] (C) -- (D) node[midway, above] {$b$};               
  \draw[->] (D) to[bend left=20] node[below] {$a,b$} (C); 

\end{tikzpicture}
\end{minipage}
\begin{minipage}{.6\textwidth}
\begin{tikzpicture}[
  vertex/.style={circle, draw=black, fill=white, minimum size=20pt, inner sep=1pt, font=\footnotesize},
  addlabel/.style={
    matrix of nodes,
    nodes={text=blue, font=\scriptsize, anchor=center},
    column sep=2pt,
    row sep=1pt,
    inner sep=0pt,
    nodes in empty cells,
  },
  >={Stealth[round]},
  shorten >=1pt,
  shorten <=1pt,
]

\node at (-1.2,0) (Invis) {}; 
\node[vertex] (v1) at (0,0) {$q_0,0,0$};
\node[vertex] (v2) at (2,0) {$q_1,5,5$};
\node[vertex] (v3) at (4,0) {$q_2,3,5$};
\node[vertex] (v4) at (6,0) {$q_0,3,1$};

\node[vertex] (v5) at (0,-2) {$q_3,3,1$};
\node[vertex] (v6) at (2,-2) {$q_2,6,6$};
\node[vertex] (v7) at (4,-2) {$q_3,6,1$};
\node[vertex] (v8) at (6,-2) {$q_0,6,1$};

\node[above=0mm of v1] {0};
\node[above=0mm of v2] {5};
\node[above=0mm of v3] {13};
\node[above=0mm of v4] {3};

\node[above=0mm of v5] {7};
\node[above=0mm of v6, label={[xshift=1.0mm, yshift=-0.3cm]14}] {};
\node[above=1mm of v7, label={[xshift=-5.0mm, yshift=-0.5cm]7}] {};
\node[above=0mm of v8] {6};



\draw[->] (Invis) -- (v1); 
\draw[->, bend left=0] (v1) to (v2);
\draw[->, bend left=0] (v2) to (v3);
\draw[->, bend left=0] (v3) to (v4);
\draw[->, bend left=25] (v4) to (v2);

\draw[->, bend right=0] (v3) to (v5);
\draw[->, bend right=0] (v5) to (v6);
\draw[->, bend right=0] (v6) to (v7);
\draw[->, bend right=20] (v7) to (v6);

\draw[->, bend right=30] (v6) to (v8);
\draw[->, bend left=10] (v8) to (v2);

\end{tikzpicture}
\end{minipage}
    \caption{An example of a D2PW (left) and the DPW constructed (right)}
    \label{fig:conj-parities}
\end{figure}
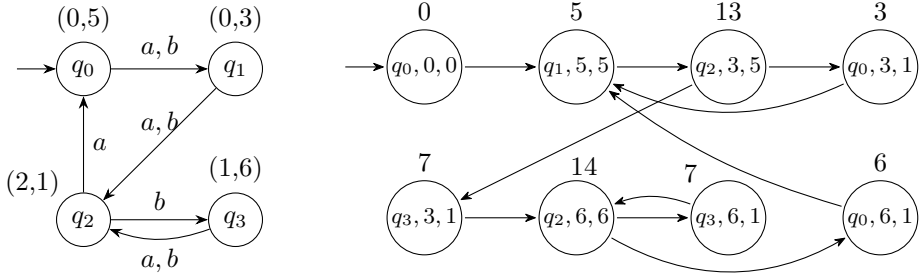
\end{longversion}

The number of priorities in $\A'$ is in $O(d_1 \cdot d_2)$ and the
blowup in the number of states is by a factor $d_2^{d_1/2}$. 
By exchanging the two priority functions in our construction, 
we obtain a blowup factor of $d_1^{d_2/2}$, and we can choose
the construction with the smaller of the two factors.

\begin{theorem}\label{theo:D2PW-DPW}
Given a D2PW $\A$, we can construct a DPW $\A'$ such that $L(\A) = L(\A')$,
the number of priorities in $\A'$ is $O(d_1 \cdot d_2)$,
and the blowup in the number of states (of $\A'$ as compared to $\A$) 
is by a factor $\sqrt{\min(d_1^{d_2},d_2^{d_1})}$.
\end{theorem}

\begin{longversion}
\begin{proof}
We first show that $L(\A) \subseteq L(\A')$.
Consider a word $w \in L(\A)$ and let $p_i$ ($i=1,2$) be 
the largest $\Omega_i$-priority visited infinitely often 
in the run of $\A$ on $w$. Then both $p_1$ and $p_2$ are even
and by construction of $\A'$, the largest priority visited infinitely often 
along the run of $w$ in $\A'$ is $p_{12} = \frac{p_1 \cdot (d_2+2)}{2} + p_2$ (see definition
of $\Omega_{12}$) because infinitely often the register $r_{p_1}$ will 
contain the $\Omega_2$-priority $p_2$, and from some point on it will never contain a larger value.
As we assumed that $d_2$ is even, we get that $p_{12}$ is even and $w \in L(\A')$.

Conversely, we show that $L(\A') \subseteq L(\A)$.
Consider a word $w \in L(\A')$ and let $p_{12}$ be the largest $\Omega_{12}$-priority visited infinitely often 
in the run of $\A'$ on $w$. Then $p_{12}$ is even and
it is easy to see that the largest $\Omega_1$-priority $p_1$ visited infinitely often 
in the run of $\A$ on $w$ is even: by contradiction, if it was odd, then the largest $\Omega_{12}$-priority
visited infinitely often would be odd as well by the properties ($P_1$) and ($P_2$).
The states with $\Omega_{12}$-priority $p_{12}$ in the run of $\A'$ correspond
to states with $\Omega_1$-priority $p_1$ in the run of $\A$, and 
$p_{12} = \frac{p_1 \cdot(d_2 + 2)}{2} + r_{p_1}$ (see definition of $\Omega_{12}$)
where $r_{p_1}$ is the content of the register for $\Omega_1$-priority $p_1$,
that is the largest $\Omega_2$-priority since the last visit to $\Omega_1$-priority $p_1$.

Hence by property ($P_3$), among the (infinitely many) visits to $p_1$ in the run of $\A$,
there are infinitely many positions where $r_{p_1}$ is $p_2 = p_{12} - \frac{p_1 \cdot(d_2 + 2)}{2}$, which is even (as we assumed that $d_2$ is even),
and from some point on, the value of $r_{p_1}$ is never greater.
It follows that, from that point on, the run of $\A$ can be decomposed (partitioned)
into finite segments whose largest $\Omega_2$-priority is at most $p_2$ and infinitely often 
equal to $p_2$. Hence the largest $\Omega_2$-priority visited infinitely often 
in the run of $\A$ on $w$ is $p_2$, which is even.
We conclude that the run of $\A$ on $w$ satisfies both parity conditions and $w \in L(\A)$.
\end{proof}
\end{longversion}

As a corollary, the transformation in Theorem~\ref{theo:D2PW-DPW} is polynomial
when either of the priority functions in the D2PW has a fixed (but arbitrary) range. In particular, it follows that the complexity of the synthesis problem for games with a conjunction of two sure (or almost-sure) parity conditions, which is coNP-complete in general,
becomes NP~$\cap$~coNP when either of the priority functions
defining the parity conditions has a fixed range,
as it reduces in polynomial time to a standard parity game.

Note that the same construction can be used to transform a disjunction of two
parity conditions into a single parity condition. 
This follows since $\psi_1 \cup \psi_2$ = $(\psi^c_1 \cap \psi^c_2)^c$ and because the complement of a parity condition defined by a priority function $\Omega$ is also a parity condition, defined by the priority function $\Omega^{+1}$.

\section{Winning Strategy and Winning Region}
We describe informally a sure-almost-sure winning strategy for Player~$1$.
We start by playing an almost-sure winning strategy for the conjunction
of parity objectives $\Parity(\Omega_1) \cap \Parity(\Omega_2)$. 
Conjunction of parity objectives can be expressed as a single parity objective
over an exponentially larger game graph (Theorem~\ref{theo:D2PW-DPW}) and therefore we can assume
that the almost-sure winning strategy uses finite memory.

To guarantee sure winning for the first objective,
the crux is then to avoid outcomes that violate $\Parity(\Omega_1)$ (which is 
possible, though with probability $0$) by switching from time to time to a strategy that ensures 
a visit to a large even priority for $\Omega_1$, ignoring the second objective thus 
possibly at the price of visiting a large odd priority for $\Omega_2$, 
then switching back to an almost-sure winning strategy for the conjunction of objectives.
The switching should be so rare that the probability of switching 
infinitely often, thereby potentially violating the second objective $\Parity(\Omega_2)$, would be $0$. 
On the other hand, the switching should still guarantee in all cases to eliminate 
the possibility of a single outcome violating $\Parity(\Omega_1)$.

Intuitively, a condition for switching is to have seen in the play a long sequence
of odd priorities, without seeing any larger priority.   
Given odd priority $p$ and integer $N$, let 
$$\Unlucky(p,N) = \{\rho \in V^{\omega}  \mid \text{the first } N \text{ elements of } \pi_{\Omega_1 \geq p}(\rho) \text{ are } p \}$$
be the event that at least $N$ priorities $\geq p$ occur and the first $N$ such priorities are $p$.\todo{Can we add an intuition as to why not consider ore than priority $p$ for the first $N$ elements?}

In the MDP obtained after fixing an almost-sure winning strategy for $\Parity(\Omega_1)$ in the game $\Game$,
it becomes unlikely to be unlucky (i.e., to observe $\Unlucky(p,N)$) as $N$ grows to infinity.

\begin{lemma}\label{lem-prob-unlucky}
Consider an MDP with a priority function 
$\Omega_1$ defining the parity objective $\psi = \Parity(\Omega_1)$,
such that $\Prob_{v}^{\sigma}(\psi) = 1$ for all vertices $v \in V$ and all strategies $\sigma$.  
Let $p$ be an odd priority in the range of $\Omega_1$.
Then for all $\epsilon > 0$, there exists $N_{\epsilon}$ such that for all strategies $\sigma$
and all vertices $v$, we have $\Prob_{v}^{\sigma}(\Unlucky(p,N_{\epsilon})) < \epsilon$.
\end{lemma}

\begin{longversion}
\begin{proof}
We prove the contrapositive. Assume that there exists $\epsilon > 0$ 
such that for all $N$ there exists a strategy $\sigma_N$ and vertex $v_N$
for which $\Prob_{v_N}^{\sigma_N}(\Unlucky(p,N)) \geq \epsilon$.

Consider a vertex $v_{\infty}$ such that $v_{\infty} = v_N$ for infinitely many 
values of $N$.  
We construct a strategy $\sigma_{\infty}$ such that for all $N$ we have 
$\Prob_{v_{\infty}}^{\sigma_{\infty}}(\Unlucky(p,N)) \geq \epsilon$,
which entails that $\Prob_{v_{\infty}}^{\sigma_{\infty}}(\Unlucky(p,\infty)) \geq \epsilon$,
and thus $\Prob_{v_{\infty}}^{\sigma_{\infty}}(\psi) \leq 1-\epsilon < 1$,
which concludes the proof of the contrapositive.

Define $\sigma_{\infty}$ inductively. Let $\sigma_{\infty}(v_{\infty}) = v_0$
where $v_0$ is such that $\sigma_N(v_{\infty}) = v_0$ for infinitely many values of $N$.  
Consider the set $\Sigma$ of all strategies $\sigma_N$ that agree with $\sigma_{\infty}$
on the value  
at $v_{\infty}$, that is $\sigma_N(v_{\infty}) = \sigma_{\infty}(v_{\infty})$. This set is infinite. 
The inductive construction is as follows: assume that we have constructed an infinite set $\Sigma_k \subseteq \Sigma$ of strategies\todo{Add that $\Sigma_k$ of $\sigma_N$ strategies?} that agree
with $\sigma_{\infty}$ at all histories of length at most $k$, and define the
value of $\sigma_{\infty}$ at all histories of length $k+1$ and an
infinite set $\Sigma_{k+1} \subseteq \Sigma_k$ of \todo{Add that $\Sigma_k$ of $\sigma_N$ strategies?} strategies that agree
with $\sigma_{\infty}$ for\todo{at or for?} all histories of length at most $k+1$. 

For each strategy $\sigma \in \Sigma_k$, consider the tuple of values of 
$\sigma$ at each history of length $k+1$. Note that this is a finite tuple since 
the MDP is finitely-branching. Hence there exists an infinite set 
$\Sigma_{k+1} \subseteq \Sigma_k$ of strategies that agree on the 
values in that tuple. We use this tuple to define the value of $\sigma_{\infty}$ at 
the histories of length $k+1$. 

This construction ensures that $\sigma_{\infty}$ agrees with some strategy $\sigma_k$
over all prefixes of length $k$, for arbitrarily large $k$. 
Let $\Unlucky^{\leq k} (p,N) = \{\rho \in V^{\omega}  \mid \text{the first } N \text{ elements of } \pi_{\Omega_1 \geq p}(\rho \lhd k) \text{ are } p \}$,
where $\tau(\lhd k)$ is the prefix of length $k$ of $\rho$, be the event that at least $N$ priorities $\geq p$ occur within the 
first $k$ vertices, and the first $N$ such priorities are $p$. 
The (non-decreasing) sequence $\Prob_{v_{\infty}}^{\sigma_{\infty}}(\Unlucky^{\leq k}(p,N))$ for $k=1,2,\dots$
converges to $\Prob_{v_{\infty}}^{\sigma_{\infty}}(\Unlucky(p,N))$,
which is also equal to $\lim_{k \to \infty} \Prob_{v_{\infty}}^{\sigma_{k}}(\Unlucky^{\leq k}(p,N)) \geq \epsilon$
since all terms in the sequence are at least $\epsilon$.
Hence $\Prob_{v_{\infty}}^{\sigma_{\infty}}(\Unlucky(p,N)) \geq \epsilon$ as required.
\end{proof}
\end{longversion}

Intuitively, since $\Unlucky(p,N)$ happens with positive probability, and
if Player $1$ switches from an almost-sure winning strategy whenever the event $\Unlucky(p,N)$ happens, using Borel-Cantelli lemma~\cite{Durrett2010}, Player $1$ will then switch (even infinitely often) with probability~$1$. 
This holds for all fixed $N$ even though the probability to be unlucky can be made arbitrarily small by choosing $N$ appropriately (according to Lemma~\ref{lem-prob-unlucky}).
It is undesirable to switch infinitely often, as this may cause the
violation of the second parity objective. 
However, if we increase $N$ upon seeing $N$ priorities $\geq p$ (whether unlucky or not), 
the probability to switch infinitely often may be reduced to~$0$, which is acceptable.

For all $i > 0$, let $N_i = N_{\epsilon}$ where  $N_{\epsilon}$ is defined by Lemma~\ref{lem-prob-unlucky} for $\epsilon = 2^{-i}$.
Using this sequence $(N_i)$, we show that the probability to never be unlucky from some round $i$ can be (lower) bounded by 
the expression $\prod_{j \geq i} \left(1-2^{-j}\right)$.
Simple calculations show that this expression tends to $1$ as $i \to \infty$ (e.g., see the proof of~\cite[Lemma~12]{BRR17}).

\subsection{Characterisation of the winning region}

\begin{figure}[!t]
   \begin{minipage}[c]{.52\linewidth}
	   \begin{center}
		\input{figures/even.tex}
	   \end{center}
	   \caption{Decomposition of a winning region with an even largest $\Omega_{1}$-priority $d$, into the sure attractor $A$ (for Player~$1$) to $d$
	   and the subgame closure of $B = V \setminus A$. \label{fig:even}}
   \end{minipage} \hfill
   \begin{minipage}[c]{.46\linewidth}
	   \begin{center}
	      \input{figures/odd.tex}
	   \end{center}
	  \caption{Decomposition of a winning region with an odd largest $\Omega_{1}$-priority $d$, into the positive attractor $A$ (for Player~$2$) to $d$
	   and the subgame induced by $B = V \setminus A$. \label{fig:odd}}
   \end{minipage}
\end{figure}

We present two lemmas giving characterisations of 
(1) the sure-almost-sure winning region~$W$ (for Player~$1$) when its largest $\Omega_1$-priority is even, and 
(2) the sure-almost-sure spoiling region~$V \setminus W$ (for Player~$2$) when its largest $\Omega_1$-priority is odd.
These are the base cases required to establish the correctness of the recursive algorithm that we present in Section~\ref{sec:alg}.
The symmetric cases, where the largest $\Omega_1$-priority is odd in the winning region
or is even in the spoiling region, are solved by recursive calls on subgames with a
smaller number of $\Omega_1$-priorities.

In Lemma~\ref{lem:char-even}, we show that the sure-almost-sure winning region~$W$,
which is a trap for Player~$2$ (inducing a subgame in which all vertices are winning),
is characterised, when its largest $\Omega_1$-priority $d$ is even, by the existence of (see also \figurename~\ref{fig:even}):
\begin{itemize}
\item an almost-sure winning strategy $\sigma_{AS}$ for Player~$1$, for objective $\Parity(\Omega_1) \cap \Parity(\Omega_2)$;
\item a sure-almost-sure winning strategy $\sigma_{sub}$ for Player~$1$
for the parity objectives $\Parity(\Omega_1)$ and $\Parity(\Omega_2)$ in the subgame closure $\closure{\Game}_{B}$,
where $B = V \setminus A$ and $A = \SureAttr_1(V(\Omega_1,d))$.
\end{itemize}

It is easy to show that these conditions are necessary for Player~$1$ to win in $W$,
and the crux of Lemma~\ref{lem:char-even} is to show that they
are sufficient as well, by constructing a sure-almost-sure winning strategy $\sigma_1$
from $\sigma_{AS}$, $\sigma_{sub}$, and the attractor strategy $\sigma_{Attr}$ in $A$.
We describe $\sigma_1$ informally when the sure winning objective is a B\"uchi
condition, namely the priorities are $\{1,2\}$: 
a sure-almost-sure winning strategy first plays like $\sigma_{AS}$,
and sets a horizon $N$ by which priority~$2$ should be visited at least once.
The larger is $N$, the more likely this would happen (Lemma~\ref{lem-prob-unlucky}). 
Nevertheless, in order to ensure the stronger objective of sure winning for $\Parity(\Omega_1)$,
the strategy $\sigma_1$ deviates from $\sigma_{AS}$ in the unlucky event that 
priority~$2$ is not visited within horizon $N$: as long as the history remains in $B$,
we follow the sure-almost-sure winning strategy $\sigma_{sub}$, and whenever 
the history enters $A$ we follow the attractor strategy $\sigma_{Attr}$ until
priority~$2$ is visited and then continue with $\sigma_{AS}$. The attractor strategy 
may endanger the (almost-sure) satisfaction of the second parity objective,
for instance by visiting a large odd $\Omega_2$-priority.
We ensure that the probability that this happens infinitely often is~$0$
by increasing the horizon $N$ as described after Lemma~\ref{lem-prob-unlucky}.
The generalization of this construction to a sure parity objective is
technical, requiring to track each $\Omega_1$-priority, but conceptually analogous.


\begin{lemma}\label{lem:char-even}
Let $\Game$ be a stochastic game with sure-almost-sure parity objective induced
by the priority functions $\Omega_1,\Omega_2$, such that the largest $\Omega_1$-priority $d$
in $\Game$ is even. Let $A = \SureAttr_1(V(\Omega_1,d))$ and let $B = V \setminus A$.

All vertices 
in $\Game$ are sure-almost-sure winning for Player~$1$ for the parity objectives 
$\Parity(\Omega_1)$ and $\Parity(\Omega_2)$ if and only if:

\begin{itemize}
\item all vertices in $\Game$ are almost-sure winning for Player~$1$ for the objective
$\Parity(\Omega_1) \cap \Parity(\Omega_2)$, and
 
\item  all vertices in $\closure{\Game}_{B}$ are sure-almost-sure winning for Player~$1$ for the parity objectives 
$\Parity(\Omega_1)$ and $\Parity(\Omega_2)$.
\end{itemize}
\end{lemma}

\begin{proof}
We first show that the conditions in the lemma are sufficient to ensure that 
all vertices in $\Game$ are sure-almost-sure winning for Player~$1$.

Let $\sigma_{AS}$ be an almost-sure winning strategy for Player~$1$ for the objective
$\Parity(\Omega_1) \cap \Parity(\Omega_2)$, let $\sigma_{Attr}$ be the sure-attractor
strategy to $V(\Omega_1,d)$, and $\sigma_{sub}$ be sure-almost-sure winning strategy for Player~$1$
for the parity objectives $\Parity(\Omega_1)$ and $\Parity(\Omega_2)$ in $\closure{\Game}_{B}$.

We define a strategy $\sigma_1$ for Player~$1$, and then show that it is
sure-almost-sure winning. 
The strategy uses a Boolean flag $\unlucky$ and integer variables $k_p, i_p$ for each odd priority $p$ occurring in $\Game$.
The variable $k_p$ is a pointer to a position in the history, and $i_p$
is the phase number associated with priority $p$. 
Initially $\unlucky = \bot$ (false) and $k_p = i_p = 0$ for all odd $p$.

Given a history $\tau \in V^* V_1$, the value $\sigma_1(\tau)$ and the update of the variables
are defined by executing the following instructions sequentially, and terminating
when an instruction \emph{play} is executed (where the sequence $(N_i)$ was defined after Lemma~\ref{lem-prob-unlucky}):
\todo{add informal description later}

\begin{enumerate}
\item if $\Omega_1(\Last(\tau)) = d$, then set $\unlucky = \bot$ and play $\sigma_{AS}(\tau)$; \label{stra-d}

\item if $\unlucky = \bot$, 
\begin{enumerate}
\item if $\abs{\pi_{\Omega_1\geq p}(\tau(k_p \rhd))} < N_{i_p}$ for all odd $p$, then play $\sigma_{AS}(\tau)$;

\item else, for each odd $p$ such that $\abs{\pi_{\Omega_1\geq p}(\tau(k_p \rhd))} \geq N_{i_p}$,
\begin{enumerate}
\item if the first $N_{i_p}$ elements of $\pi_{\Omega_1 \geq p}(\tau(k_p \rhd))$ are all $p$'s, then set $\unlucky = \top$; \label{stra-unlucky}
\item increment $i_p$, and set $k_p = \abs{\tau}$; 
\end{enumerate}

\end{enumerate}

\item if $\unlucky = \bot$, then play $\sigma_{AS}(\tau)$;

\item if $\unlucky = \top$, \label{stra-if-unlucky}
\begin{enumerate}
\item if $\Last(\tau) \in B$, then play $\sigma_{sub}(\tau')$ \label{stra-sub}
where $\tau'$ is the longest suffix of $\tau$ that is entirely in $B$;
\item else $\Last(\tau) \in A$, and play $\sigma_{Attr}(\tau)$. \label{stra-attr}
\end{enumerate}

\end{enumerate}

Note that a play of $\Game$ that remains in $B$ is also a play in $\closure{\Game}_{B}$, hence using 
the strategy $\sigma_{sub}$ in Condition~\ref{stra-sub} is well-defined.
\begin{shortversion}
We show in the appendix that $\sigma_1$ is sure-almost-sure winning.
\end{shortversion}

\begin{longversion}
We show that $\sigma_1$ is sure-almost-sure winning.
First show that $\sigma_1$ is sure winning for objective $\Parity(\Omega_1)$.
Consider an arbitrary outcome $\rho$ of $\sigma_1$, let $p_{\max}$ be the largest 
priority (according to $\Omega_1$) occurring infinitely often in $\rho$,
and consider a position $\hat{k}$ in $\rho$ after which no priority larger than $p_{\max}$ occurs.
We show that $p_{\max}$ is even. We proceed by contradiction, assuming that $p_{\max}$ is odd (hence $p_{\max} < d$).
Then along $\rho$, the variable $\unlucky$ must be infinitely often equal to $\top$, since condition~\ref{stra-unlucky} holds
infinitely often (after position $\hat{k}$, that is, when $k_{p_{\max}} > \hat{k}$).
Condition~\ref{stra-d} never holds after position $\hat{k}$,   
Condition~\ref{stra-if-unlucky} holds infinitely often. After position $\hat{k}$,
Condition \ref{stra-attr} does not hold, because otherwise priority $d > p_{\max}$
would occur. Hence Condition \ref{stra-sub} always holds after position $\hat{k}$
and thus the suffix of $\rho$ is an outcome of the sure-almost-sure winning
strategy $\sigma_{sub}(\tau')$, hence $p_{\max}$ is even and not odd,
which establishes the contradiction. Hence  $p_{\max}$ is even and $\sigma_1$ 
is sure winning for objective $\Parity(\Omega_1)$.

Second, we show that $\sigma_1$ is almost-sure winning for objective $\Parity(\Omega_2)$.
Note that, under the event that always $\unlucky = \bot$ from some point on,
the strategy $\sigma_1$ plays like $\sigma_{AS}$, and thus the parity objective 
$\Parity(\Omega_2)$ holds with probability~$1$ thanks to prefix independence
of parity objectives. To complete the proof, we show that the event that $\unlucky = \bot$ always holds
from some point on has probability~$1$. By the choice of $N_i$ (see Lemma~\ref{lem-prob-unlucky}),
for all strategies of Player~$2$
the probability that $\unlucky = \bot$ always holds after round $i$ is at least 
$\prod_{j \geq i} \left(1-2^{-j}\right)$, which tends to $1$ as $i \to \infty$ (e.g., see the proof of~\cite[Lemma~12]{BRR17}).
\end{longversion} 

We now prove the converse. Consider a game $\Game$ where
all vertices are sure-almost-sure winning for Player~$1$ for the parity objectives 
$\Parity(\Omega_1)$ and $\Parity(\Omega_2)$. First, since sure winning implies almost-sure
winning, the first condition holds. Second, a sure-almost-sure winning strategy for Player~$1$ 
can be played in $\closure{\Game}_{B}$, and either the play always remains in $B$, 
or it eventually stays in the winning vertex $v_{\sf sink}$. 
In both case, Player~$1$ is sure-almost-sure winning for the parity objectives, 
showing that the second condition holds.
\end{proof}

The sure-almost-sure spoiling region~$V \setminus W$ may not induce a subgame
in general, since probabilistic vertices may `leak' in the winning region. 
Lemma~\ref{lem:char-odd} gives a characterisation of the spoiling region
in the case that it induces a subgame, which is sufficient for proving the
correctness of our algorithm. 

If Player~$2$ has a spoiling strategy (from every initial vertex) in the game 
obtained from a game $\Game$ with an odd largest $\Omega_{1}$-priority $d$ by removing the 
positive attractor (for Player~$2$) to $d$\todo{Some rephrasing might be better: obtained by removing ... from a game $\Game$ \dots}, then he has a spoiling strategy in the original game $\Game$ as well (see also \figurename~\ref{fig:odd}).
The standard composition a spoiling strategy in the subgame with a (positive) attractor strategy gives a spoiling strategy in the original game, as shown in the proof of Lemma~\ref{lem:char-odd}.

\begin{lemma}\label{lem:char-odd}
Let $\Game$ be a stochastic game with sure-almost-sure parity objective induced
by the priority functions $\Omega_1,\Omega_2$, such that the largest $\Omega_1$-priority $d$
in $\Game$ is odd. Let $A = \PosAttr_2^\Game(V(\Omega_1,d))$ and let $B = V \setminus A$.

The sure-almost-sure winning region for Player~$1$ in $\Game$ for the parity objectives 
$\Parity(\Omega_1)$ and $\Parity(\Omega_2)$ is empty if and only if
the sure-almost-sure winning region for Player~$1$ in $\Game_{\upharpoonright B}$ for the same 
objectives is empty.
\end{lemma}

\begin{longversion}
\begin{proof}
We first show the sure-almost-sure winning region for Player~$1$ being
empty in $\Game_{\upharpoonright B}$ is sufficient for it to be empty in $\Game$.
By the premise, given an arbitrary strategy $\sigma_1$ for Player~$1$ in $\Game_{\upharpoonright B}$
and a vertex $v$ in $\Game_{\upharpoonright B}$,
there exists a strategy $\sigma_2$ for Player~$2$ such that 
either $\Outcome(v,\sigma_1,\sigma_2) \cap \Parity^\textsf{c}(\Omega_1) \neq \emptyset$ 
or $\Prob_{v}^{\sigma_1,\sigma_2}(\Parity^\textsf{c}(\Omega_2)) > 0$, 
We call $\sigma_2$ a \emph{spoiling} strategy against $\sigma_1$, omitting 
the initial vertex because the same strategy $\sigma_2$ can be used from all vertices
(sometimes calling it simply a spoiling strategy, when $\sigma_1$ is clear from the context).

Given an arbitrary strategy $\sigma_1$ for Player~$1$ in $\Game$, we construct 
a spoiling strategy $\sigma_2$ for Player~$2$ as follows, for all histories $\tau \in V^* V_2$:
\begin{itemize}
\item if $\Last(\tau) \in V(\Omega_1,d)$, then $\sigma_2(\tau)$ is defined arbitrarily,
\item otherwise, if $\Last(\tau) \in A$ is in the positive attractor, then $\sigma_2(\tau) = \sigma_{Attr}(\tau)$
where $\sigma_{Attr}$ is the positive-attractor strategy (for Player~$2$) to $V(\Omega_1,d)$,
\item otherwise $\Last(\tau) \in B$; let $\rho$ be the longest (possibly empty)
prefix of $\tau$ such that $\Last(\rho) \notin B$ and let $v_0 \in B$ such that $\rho v_0$
is a prefix of $\tau$; consider the strategy $\sigma_{\rho}$
of Player~$1$ relative to $\rho$, defined by $\sigma_{\rho}(\tau') = \sigma_1(\rho \tau')$
for all $\tau' \in \{v_0\} B^* (B \cap V_1)$ with initial vertex $v_0$,
which is a strategy in $\Game_{\upharpoonright B}$. Define $\sigma_2(\tau)$
as what is played by a spoiling strategy of Player~$2$ against $\sigma_{\rho}$
at the history $\tau'$ such that $\rho \tau' = \tau$.
\end{itemize}

We show that $\sigma_2$ is a spoiling strategy against $\sigma_1$.
Consider the set $\Outcome(v,\sigma_1,\sigma_2)$ of outcomes, and two
possible cases:

\begin{itemize}
\item if there exists an outcome 
that visits the positive attractor $A$ infinitely often, then it also
visits the largest odd $\Omega_{1}$-priority $d$ infinitely often, violating the first
parity objective $\Parity(\Omega_1)$;
\item otherwise, 
there exists a prefix $\tau$ of an outcome such that in all outcomes
extending $\tau$, the suffix after $\tau$ never visits $A$; hence the history
$\tau$ occurs with a positive probability, and by the construction of $\sigma_2$ 
the play after $\tau$ is an outcome of a spoiling strategy of Player~$2$ in $\Game_{\upharpoonright B}$,
which ensures that, either there is an outcome that violates the first
parity objective $\Parity(\Omega_1)$, or the probability to violate
the second parity objective $\Parity(\Omega_2)$ is positive. 
By prefix-independence of parity objectives, the same holds from the initial 
vertex in $\Game$.
\end{itemize}

In both cases, the strategy $\sigma_2$ is a spoiling strategy against $\sigma_1$,
which concludes the first part of the proof.

To show that the condition in the lemma is necessary, recall that $B$ is a trap for \PT and thus a strategy for Player~$1$ in $\Game_{\upharpoonright B}$ is also a strategy in $\Game$,
and a spoiling strategy defined in $\Game$ can be immediately restricted
to a spoiling strategy in $\Game_{\upharpoonright B}$\todo{This is true since $B$ is a trap for \PT; may be we can add this.}.
\end{proof}
\end{longversion}

\section{Algorithm}\label{sec:alg}

We present an algorithm for the sure-almost-sure synthesis problem that computes the
winning region of Player~$1$ by a recursive procedure in the style of Zielonka's algorithm~\cite{Zie98}.
The base case is when the game is empty or contains only the sink vertex $v_{\sf sink}$. 
Then, two cases are possible:
\begin{itemize}
\item if the largest $\Omega_{1}$-priority $d$ in $V$ is even, 
the first step is to compute the almost-sure winning region $W_{AS}$ for the conjunction 
$\Parity(\Omega_1) \cap \Parity(\Omega_2)$ of parity objectives, which 
is a trap for Player~$2$ and an over-approximation of the sure-almost-sure winning region.
This can be done by expressing the conjunction of parity objectives as a single parity objective
over an exponentially larger game graph (Theorem~\ref{theo:D2PW-DPW}), and using standard algorithms
for almost-sure parity~\cite{CJH03,CJH04}. 
In the subgame $\Game_{\upharpoonright W_{AS}}$ induced by $V = W_{AS}$, 
compute the sure attractor $A$ for Player~$1$ to the vertices with priority~$d$, then solve (recursively) 
the subgame closure $\closure{\Game}_{V \backslash A}$. If all vertices in the 
subgame closure are winning for Player~$1$ (line~\ref{algo:sas-all-win}), then by the characterisation 
of Lemma~\ref{lem:char-even} all vertices in $\Game$ are also winning for Player~$1$.
Otherwise, some set $B \subseteq V \setminus A$ of vertices in the subgame closure is not winning for Player~$1$,
and by the key property of subgame closures\todo{Add this as a proposition?}, those vertices are also not winning in $\Game$.
Removing the positive attractor for Player~$2$ to $B$, it remains to solve (recursively)
the subgame $\Game_{\upharpoonright V \backslash B}$ (line~\ref{algo:sas-subgame-even}).
\item if the largest $\Omega_{1}$-priority $d$ in $V$ is odd, the algorithm proceeds analogously,
computing the positive attractor $A$ for Player~$2$ to the vertices with priority~$d$,
then solving (recursively) the subgame $\Game_{\upharpoonright V \backslash A}$, 
using the characterisation in Lemma~\ref{lem:char-odd} if no vertex is winning for 
Player~$1$ (line~\ref{algo:sas-all-lose-return}), or the key property of traps and 
subgame closures otherwise (line~\ref{algo:sas-recursive-return}).
\end{itemize}

\begin{algorithm}[t]
    \caption{{\sf solve}($\Game$) to compute the sure-almost-sure winning region for parity objectives}\label{algo:sas}
    \begin{algorithmic}[1]
        \Require{Stochastic two-dimensional parity game given by $\Game = (V_1, V_2, E)$ and $(\Omega_1, \Omega_2)$}
        \Ensure{Sure-almost-sure winning region $W_1$, and losing region $W_2$, for the parity objectives $\Parity(\Omega_1)$ and $\Parity(\Omega_2)$.}
        \If{$V=\emptyset$ or $V=\{v_{\sf sink}\}$}{ \Return V}\EndIf \label{algo:sas-base-case}
        \State $d = \max_{v \in V} \Omega_1(v)$
            \If{$d$ is even}
                \State $\Game = \Game_{\upharpoonright W_{AS}}$  \algorithmiccomment{ $W_{AS}$ is  the almost-sure winning region for $Parity(\Omega_1) \cap \Parity(\Omega_2)$. \label{algo-AS}}
                \State $V = W_{AS}$ \label{algo-newV}
                \State $Z = \{v \in V \mid \Omega_1(v) = d\}$ \label{algo:defZ} \algorithmiccomment{$Z$ may be empty.}
                \State $A = \SureAttr_1(Z)$ \label{algo:defA}
                \State $W'_1, W'_2 = {\sf solve}(\closure{\Game}_{V \backslash A})$ \label{Line:P1trap} 

                \If{$W'_2 = \emptyset$ \label{algo:sas-all-win}}
                    \State \Return $W_1 = V, W_2 = \emptyset$
                \Else
                    \State $B = \PosAttr_2(W'_2)$
                    \State $W_1'', W_2'' = {\sf solve}(\Game_{\upharpoonright V \backslash B})$  \label{algo:sas-subgame-even}
                    \State \Return $W_1 = W_1'', W_2 = W_2'' \cup B$
                \EndIf
            \Else \label{algo:main-else}
                \State $Z = \{v \in V \mid \Omega_1(v) = d\}$  \algorithmiccomment{$Z$ is nonempty.}
                \State $A = \PosAttr_2(Z)$
                \State $W'_1, W'_2 = solve(\Game_{\upharpoonright V \backslash A})$ \label{algo:defW1prime}

                \If{$W'_1 = \emptyset$ \label{algo:sas-all-lose}}
                    \State \Return $W_1 = \emptyset, W_2 = V$ \label{algo:sas-all-lose-return}
                \Else \label{algo:odd-else}
                    \State $B = \SureAttr_1(W'_1)$ \label{algo:defB}
                    \State $W_1'', W_2'' = {\sf solve}(\closure{\Game}_{V \backslash B})$ \label{algo:sas-subgame-odd} 
                    \State \Return $W_1 = B \cup W_1'' \setminus \{v_{\sf sink}\}, W_2 = W_2''$ \label{algo:sas-recursive-return}
                \EndIf
            \EndIf
    \end{algorithmic}
\end{algorithm}


We prove the correctness of Algorithm~\ref{algo:sas} by induction on the size of 
the game, defined as the number $\abs{V \setminus \{v_{\sf sink}\}}$ of vertices 
excluding the sink vertex.
Consider a stochastic game $\Game = \tuple{V, E, \delta}$, and in the base case 
where there is no vertices other than $v_{\sf sink}$, the winning region is $V$ 
(line~\ref{algo:sas-base-case} of Algorithm~\ref{algo:sas}).

Now let the number of vertices other then $v_{\sf sink}$ be $n+1$ ($n\geq 0$) and assume Algorithm~\ref{algo:sas} 
computes the winning regions in all games of size at most $n$.
Let $d$ be the highest $\Omega_1$-priority assigned to a vertex in $V$.

First, consider the case that $d$ is even.
The algorithm removes the vertices that are not almost-sure winning 
for the conjunction $\Parity(\Omega_1) \cap \Parity(\Omega_2)$ of parity objectives,
as they cannot be sure-almost-sure winning. Then let $V = W_{AS}$ (line~\ref{algo-newV}),
let $Z \subseteq V$ be the set of all vertices with priority $d$, and let $A = \SureAttr_1(Z)$.
Since $Z \neq \emptyset$ and $Z \subseteq A$, the set $V \backslash A$ has strictly fewer vertices (other than the sink) than $V$
and by the induction hypothesis we get that $W_1'$ (line~\ref{Line:P1trap}) is the winning region of Player~$1$ 
in the subgame closure $\closure{\Game}_{V \backslash A})$ (and $W_2' = V \setminus W_1'$ that of Player~$2$).
We now consider the following two cases.
\begin{enumerate}
\item if $W_2' = \emptyset$, then $W_1 = V$ by the characterisation of Lemma~\ref{lem:char-even}, 
that is Player~$1$ wins from every vertex in $\Game$.

\item if $W_2' \neq \emptyset$, then let $B = \PosAttr_2(W_2')$ in $\Game$,
and since $v_{\sf sink} \not\in W_2'$ we can apply the induction hypothesis 
to conclude that $W_1''$ is the winning region of Player~$1$ in the subgame
$\Game_{\upharpoonright V \backslash B}$ (line~\ref{algo:sas-subgame-even}).
Since $V \backslash B$ is a trap for Player~$2$, the key property of traps 
(see Section~\ref{sec:subgames_and_traps}) ensures that $W_1''$, returned by the algorithm, is also
winning for Player~$1$ in the game $\Game$. 
\end{enumerate}

We now consider the case that $d$ is odd.
Thus the set $Z \subseteq V$ all vertices with priority $d$ is nonempty, 
and let $A = \PosAttr_2(Z) \neq \emptyset$.
By the induction hypothesis, the set $W_1'$ (line~\ref{algo:defW1prime}) is the winning region of Player~$1$ 
in the subgame $\Game_{\upharpoonright V \backslash A}$.
We now consider the following two cases.

\begin{enumerate}
\item if $W_1' = \emptyset$, then $W_1 = \emptyset$ by the characterisation of Lemma~\ref{lem:char-odd}, 
that is Player~$1$ wins from nowhere in $\Game$.

\item if $W_1' \neq \emptyset$, then let $B = \SureAttr_1(W_1')$ in $\Game$.
We can apply the induction hypothesis 
to conclude that $W_1''$ is the winning region of Player~$1$ in the subgame
closure $\closure{\Game}_{V \backslash B}$ (line~\ref{algo:sas-subgame-odd}).
To show that $B \cup W_1''$ is the winning region of Player~$1$ in $\Game$,
note that $W_1'$ is a trap for Player~$2$ in the subgame $\Game_{\upharpoonright V \backslash A}$
and that $B$ is a trap for Player~$2$ in the game $\Game$.
Consider the strategy of Player~$1$ defined as follows: 
$(1)$ in $W_1'$, play like a sure-almost-sure winning strategy 
in the subgame $\Game_{\upharpoonright V \backslash A})$ (which exists and is also
a sure-almost-sure winning strategy in $\Game$); $(2)$ in $B \setminus W_1'$, play the (sure)
attractor strategy to reach $W_1'$; $(3)$ in $W_1''$, play like 
a sure-almost-sure winning strategy 
in the subgame $\Game_{\upharpoonright V \backslash A})$ (which is well-defined 
in the game $\Game$).
It is immediate that this strategy is sure-almost-sure winning from $B$
and we show that it is also the case from $W_1''$. To show sure winning for the
first parity objective, consider an outcome of the strategy from $W_1''$, and 
either the outcome remains in $W_1''$ and thus satisfies the parity condition by $(3)$,
or the outcome eventually leaves $W_1''$, thus reaching $B$ from where it 
follows some path that satisfies the parity condition by $(1)$ and $(2)$. 
Prefix-independence ensures that the outcome from $W_1''$ also satisfies the parity condition.
To show almost-sure winning for the second parity objective, the argument is 
similar, considering the conditional probability that the play eventually remains in $W_1''$
or leaves $W_1''$ and showing that in both cases,
the parity objective is satisfied almost-surely.
\end{enumerate}

The structure of Algorithm~\ref{algo:sas} is similar to Zielonka's algorithm for 
parity games~\cite{Zie98}, and so is the complexity analysis. The main 
difference is that we also need to solve stochastic games with a conjunction of 
almost-sure parity conditions (line~\ref{algo-AS}). Using the construction
in Section~\ref{sec:conjParity}, and denoting by $\ASPG(N,d)$ the time complexity 
of solving almost-sure parity games with $N$ vertices and largest priority $d$,
this can be done in time $\ASPG(N,d_1\cdot d_2)$ where $N = n \cdot \sqrt{\min(d_1^{d_2},d_2^{d_1})}$
and $n$ is the number of vertices in the game.

Let $T_n^{d_1}$ be the running time of Algorithm~\ref{algo:sas} on a game with $n$ vertices and 
largest $\Omega_1$-priority $d_1$. The running time also depends on the 
largest $\Omega_2$-priority $d_2$, but we keep this parameter implicit,
as it does not change through the recursive calls.

The crux of the analysis is to see that in the recursive calls $\closure{\Game}_{V \backslash A}$ (line~\ref{Line:P1trap} and~\ref{algo:defW1prime}),
the largest $\Omega_1$-priority in the game $\closure{\Game}_{V \backslash A}$ is at most $d_1 - 1$,
and in the recursive calls ${\sf solve}(\closure{\Game}_{V \backslash B})$  (line~\ref{algo:sas-subgame-even} and~\ref{algo:sas-subgame-odd}), the number of vertices in the game $\closure{\Game}_{V \backslash B}$ is at most $n - 1$.
Hence there is at most $n$ recursive calls for a given largest $\Omega_1$-priority $d_1$.

Besides the recursive calls and the solution of almost-sure parity games,
the other computation can be done in polynomial time (such as computing attractors),
thus bounded by $\ASPG(N,d)$. 
Hence we obtain the condition $T_n^{d_1} \leq n \cdot (T_{n-1}^{d_1} + \ASPG(N,d_1\cdot d_2))$,
which gives the (crude) bound $T_n^{d_1} \in O(n^{d_1} \cdot \ASPG(N,d_1\cdot d_2))$.

\section{Special Cases}

We derive from Algorithm~\ref{algo:sas} several results for subclasses of
the sure-almost-sure synthesis problem, giving a refined and more complete view of the problem.
The results are presented in Table~\ref{tab:memory_complexity}
where the rows (resp., columns) denote the sure (resp., almost-sure) satisfaction 
for B\"uchi, coB\"uchi, parity with fixed number of priorities (${\sf Parity^f}$), and parity condition.
For each cell, we give $(1)$ the smallest class of strategies that are sufficient for sure-almost-sure winning,
according to the classification types of memoryless, finite-memory, and infinite-memory 
strategies, and $(2)$ the complexity of solving the decision problem. 
All results are tight, in particular the complexity bounds: either the lower and upper bound 
match (completeness result), or the upper bound is the same as in a special case of
the problem for which no matching lower bound is known, 
such as games with a (single) parity objective (with an NP~$\cap$~coNP bound).

The new results in Table~\ref{tab:memory_complexity} appear in blue. 
The other results are either known from previous works, or trivially follow from the new results.


\begin{table}[!t]
\begin{center}
\caption{Memory requirement ($0$ stands for memoryless) and complexity bound for the special cases. New results appear in blue.}
\label{tab:memory_complexity}
\begin{tabular}{|l|r@{, }l|r@{, }l|r@{, }l|r@{, }l|}
\hline
\textbf{\begin{tabular}[c]{@{}l@{}}$\as, \longrightarrow$\\ $\sure, \downarrow$\end{tabular}} & 
\multicolumn{2}{l|}{\textbf{B\"uchi}}                                               & 
\multicolumn{2}{l|}{\textbf{coB\"uchi}}                                            & 
\multicolumn{2}{l|}{$\mathbf{Parity^f}$}                                              & 
\multicolumn{2}{l|}{\textbf{Parity}}                                                  \\ \hline
\textbf{B\"uchi}                                                                                    & finite & $\mathrm{P}$                                     & $\infty$~\cite{BRR17} & $\mathrm{P}$ & $\infty$ & $\mathrm{P}$                                      & $\infty$ & $\NP \cap \coNP$                                  \\ \hline
\textbf{coB\"uchi}                                                                                  & 0 & $\mathrm{P}$                                          & 0 & $\mathrm{P}$                                         & 0 & $\mathrm{P}$                                             & {\color{blue}0} & $\NP \cap \coNP$        \\ \hline
{$\mathbf{Parity^f}$}                                                                                & finite & $\mathrm{P}$                                     & $\infty$ & $\mathrm{P}$                                  & $\infty$ & {\color{blue}$\mathrm{P}$}     & $\infty$ & {\color{blue}$\NP \cap \coNP$} \\ \hline
\textbf{Parity}                                                                                    & {\color{blue}finite} & $\NP \cap \coNP$ & $\infty$ & $\NP \cap \coNP$                              & $\infty$ & {\color{blue}$\NP \cap \coNP$} & $\infty$ & $\coNP$-C~\cite{CP19}        \\ \hline
\end{tabular}
\end{center}
\end{table}

It is known that the general synthesis problem for sure-almost-sure winning is coNP-complete~\cite{CP19}
and that infinite memory is necessary for the conjunction of sure B\"uchi
and almost-sure coB\"uchi, already in MDPs~\cite{BRR17}. It is immediate from the running-time
analysis of Algorithm~\ref{algo:sas} that the problem is in P when 
both $d_1$ and $d_2$ are constant. The remaining results of 
Table~\ref{tab:memory_complexity} are established below.

\subsection{One of the parity objectives has a fixed number of priorities}

We show that the sure-almost-sure synthesis problem can be solved in NP when
one of the two parity objectives has a fixed number $d$ of priorities. 
Formally, an instance of the problem is now a stochastic game $\Game$ 
and two priority functions $\Omega_1: V \to \Nat$ and $\Omega_2: V \to \{0,\dots,d\}$, or vice versa, that is, $\Omega_1: V \to \{0,\dots,d\}$ and $\Omega_2: V \to \Nat$.
An NP algorithm is to guess the winning region $W_1$ of Player~$1$ and 
some of the sets that Algorithm~\ref{algo:sas} would compute when executed on $\closure{\Game}_{W_1}$ as a witness and verify that the guess is correct. 
Executing Algorithm~\ref{algo:sas} on a game where all vertices are winning sounds useless (we already know that all vertices are winning anyways), but it will be essential in order to define and prove the existence of certain sets (used by Algorithm~\ref{algo:sas} in intermediate computation) that we use to construct certificates in our NP algorithm.

We now give some intuition of how to define
certificates of correctness for a guessed (winning) region $W_1$. The certificates are defined
inductively for $W_1$ with largest $\Omega_1$-priority $d$, based on 
certificates for subsets of $W_1$ with largest $\Omega_1$-priority at most $d-1$ (following the recursive structure
of Algorithm~\ref{algo:sas}).
The formal proof shows that the inductive certificates can be checked in polynomial time.
First the region $W_1$ must be a trap for Player~$2$ (which can be checked in $O(E)$, 
and ensures that the subgame $\closure{\Game}_{W_1}$  is well defined).   

\begin{figure}[!t]
   \begin{minipage}[c]{.48\linewidth}
	   \begin{center}
		\input{figures/cert-even.tex}
	   \end{center}
	   \caption{Decomposition in a certificate for winning region $W_1$ with largest priority even. \label{fig:cert-even}}
   \end{minipage} \hfill
   \begin{minipage}[c]{.48\linewidth}
	   \begin{center}
	      \input{figures/cert-odd.tex}
	   \end{center}
	  \caption{Decomposition in a certificate for winning region $W_1$ with largest priority odd. \label{fig:cert-odd}}
   \end{minipage}
\end{figure}

\begin{enumerate}
\item If the largest $\Omega_1$-priority $d$ in $W_1$ is even (see \figurename~\ref{fig:cert-even}), the certificate consists
of $(1)$ a certificate that the region $W_1$ is almost-sure winning  
for the conjunction $\Parity(\Omega_1) \cap \Parity(\Omega_2)$ of parity objectives (such a 
certificate exists and can be verified in polynomial time by the remark at the end of Section~\ref{sec:conjParity}),
and $(2)$ (inductively) a certificate that the region $W_1 \setminus A$
is sure-almost-sure winning, where $A = \SureAttr_1(\{v \in W_1 \mid \Omega_1(v) = d\})$
(see also lines~\ref{algo:defZ},\ref{algo:defA} in Algorithm~\ref{algo:sas}).
The inductive certificate 
for $W_1 \setminus A$ exists by the characterisation given in Lemma~\ref{lem:char-even}.  
\item 
If the largest $\Omega_1$-priority $d$ in $W_1$ is odd (see \figurename~\ref{fig:cert-odd}), the certificate consists 
of $(1)$ a partition $U_1,\dots,U_k$ of $W_1$ into $k \leq \abs{W_1}$ blocks;   
$(2)$ nonempty sets $R_1,\dots,R_k$ such that 
$R_i \subseteq U_i \setminus \{v \in W_1 \mid \Omega_1(v) = d\}$
is a trap for Player~$2$ in 
$\closure{\Game}_{W_1 \setminus (U_1 \cup \dots \cup U_{i-1})}$
that contains no vertex with $\Omega_1$-priority $d$ and such that
$U_i = \SureAttr_1(R_i)$ (defined in the subgame 
$\closure{\Game}_{W_1 \setminus (U_1 \cup \dots \cup U_{i-1})}$), for all $1 \leq i \leq k$; and
$(3)$ (inductively) $k$ certificates showing that each region 
$R_i$ ($i=1,\dots,k$) is sure-almost-sure winning in the game $\closure{\Game}_{W_1 \setminus (U_1 \cup \dots \cup U_{i-1})}$.
Here, the certificate corresponds to the sets computed by Algorithm~\ref{algo:sas}
in the subgame $\closure{\Game}_{W_1}$ as follows: the set $R_1$ is the value of $W_1'$
(line~\ref{algo:defW1prime} of Algorithm~\ref{algo:sas}), the set $U_1$ is the value of $B$ 
(in the else-clause at line~\ref{algo:defB} of Algorithm~\ref{algo:sas}), and the sets $R_i$ and $U_i$ for $i=2,\dots,k$
are obtained in a similar in way in the $(i-1)$-th recursive call (at line \ref{algo:sas-subgame-odd}
of Algorithm~\ref{algo:sas}). Note that the set $U_1$ is well defined because 
the then-clause at line~\ref{algo:sas-all-lose-return} is never executed 
when the entire state space $V = W_1 \neq \emptyset$ is sure-almost-sure winning for Player~$1$, 
since if $W_1'$ was empty, then $W_1$ would be empty by the characterisation
given in Lemma~\ref{lem:char-odd}. 
\end{enumerate}

\begin{theorem}
The sure-almost-sure synthesis problem can be solved in NP when
one of the two parity objectives has a fixed number $d$ of priorities.
\end{theorem}

\begin{longversion}
\begin{proof}
We define an NP algorithm for the sure-almost-sure synthesis problem
with time complexity $C(n,d_1,d_2) \in O( (d_1+1) \cdot (n^2 + \mathsf{AS2P}(n,d_1,d_2)))$,
where $n = \abs{V}$ is the size of the state space and $d_i$ is the largest $\Omega_i$-priority,
and $\mathsf{AS2P}(n,d_1,d_2)$ is the (polynomial, see end of Section~\ref{sec:conjParity}) time 
complexity of solving the almost-sure synthesis problem for the conjunction of 
the parity objectives, using a nondeterministic algorithm.

Our algorithm and the proof are established by induction on the number of $\Omega_1$-priorities.
Given a game $\Game$, priority functions $\Omega_1$, $\Omega_2$, and initial vertex $v_0$,
the algorithm guesses a set $W_1$ containing $v_0$  
and verifies that $W_1$ is a trap for Player~$2$, which can be done in $O(n^2)$.
From now on, we only consider the subgame $\Game_{\upharpoonright W_1}$ 
and the algorithm will verify that all its vertices are sure-almost-sure winning. 

The base case is when the largest $\Omega_1$-priority $d_1$ is $0$. The sure parity 
objective is then trivially satisfied, and the algorithm needs to check
that Player~$1$ is almost-sure winning for the second parity objective,
which can be done in (nondeterministic) time complexity at most $\mathsf{AS2P}(n,0,d_2)$.
Therefore the time complexity of the base case is $C(n,0,d_2) \in O(n^2 + \mathsf{AS2P}(n,0,d_2))$.

The inductive case is when the largest $\Omega_1$-priority in $W_1$ is $d_1$ and we assume that
the complexity bound holds for all games with largest $\Omega_1$-priority $d_1 - 1$.

If $d_1$ is even, then by the characterisation of Lemma~\ref{lem:char-even}, 
the algorithm can guess and check a certificate witnessing 
that the region $W_1$ is a trap for Player~$2$ that is almost-sure winning for 
the conjunction $\Parity(\Omega_1) \cap \Parity(\Omega_2)$ of parity objectives,
and that the region $W_1 \setminus A$ where $A = \SureAttr_1(\{v \in W_1 \mid \Omega_1(v) = d\})$
is sure-almost-sure winning. 
The sure attractor can be computed in quadratic time,
and therefore the verification can be done in time \todo{$\mathsf{AS2P}(n,d_1,d_2)))$ is in $\NP$. Need to rephrase that the verification is in deterministic polynomial time.} at 
most $O(n^2 + \mathsf{AS2P}(n,d_1,d_2) + C(n,d_1 - 1,d_2))$, which is 
$O((d_1+1) \cdot (n^2 + \mathsf{AS2P}(n,d_1,d_2)))$ using the induction hypothesis.

If $d_1$ is odd, then the algorithm guesses 
a partition of $W_1$ into nonempty sets $U_1,\dots,U_k$ (thus $k \leq \abs{W_1}$)
and nonempty sets $R_1,\dots,R_k$ such that 
$R_i \subseteq U_i \setminus \{v \in W_1 \mid \Omega_1(v) = d_1\}$   
is a trap for Player~$2$ in 
$\Game_{\upharpoonright W_1 \setminus (U_1 \cup \dots \cup U_{i-1})}$
and such that $U_i = \SureAttr_1(R_i)$ (where the sure attractor is defined in the subgame 
$\Game_{\upharpoonright W_1 \setminus (U_1 \cup \dots \cup U_{i-1})}$), for all $1 \leq i \leq k$.
Being a partition, a trap, or a sure attractor can all be checked in quadratic time.
Moreover, the algorithm checks that each region $R_i$ ($i=1,\dots,k$) is sure-almost-sure 
winning in the game $\Game_{\upharpoonright W_1 \setminus (U_1 \cup \dots \cup U_{i-1})}$,
which can be done in (nondeterministic) time complexity at most 
$C(n_i,d_1-1,d_2) \in O(d_1 \cdot (n_i^2 + \mathsf{AS2P}(n_i,d_1,d_2)))$
by the the induction hypothesis, where $n_i = \abs{R_i}$ is the size of the traps. 
The total time needed in this case is at most 
$O\left(\sum_{i=1}^{k} n_i^2 + C(n_i,d_1-1,d_2)\right) \in O\left(\left( \sum_{i=1}^{k} n_i \right)^2 + C(\sum_{i=1}^{k} n_i,d_1-1,d_2)\right)$, 
which is in $O(n^2 + C(n,d_1-1,d_2))$ since $\sum_{i=1}^{k} n_i \leq n$, 
and therefore in $O((d_1+1) \cdot (n^2 + \mathsf{AS2P}(n,d_1,d_2)))$
as claimed.

The existence of the sets $U_1,\dots,U_k$ and $R_1,\dots,R_k$ is established
by considering the execution of Algorithm~\ref{algo:sas} on the subgame $\Game_{\upharpoonright W_1}$.
Note that the else-branch at line~\ref{algo:main-else} is executed because the
largest $\Omega_{1}$-priority in the game $\Game_{\upharpoonright W_1}$ is $d_1$, which is odd. 
The result of the recursive call at line~\ref{algo:defW1prime} defines the set $R_1 = W_1'$,
which is nonempty, otherwise $W_1 = \emptyset$ by the characterisation of Lemma~\ref{lem:char-odd}.
Hence the else-branch at line~\ref{algo:odd-else} is executed and defines the set $U_1 = B$ (line~\ref{algo:defB}).
We argue similarly in the recursive calls at line~\ref{algo:sas-subgame-odd},
defining $R_2,R_3,\dots$ and $U_2,U_3,\dots$ until the largest $\Omega_{1}$-priority in the 
game $\closure{\Game}_{V \backslash B}$ given as input to the recursive call
at line~\ref{algo:sas-subgame-odd} is less than $d_1$. Then it must be that
$W_1'' = V \backslash B$ and this set is closed under sure attractor for Player~$1$ 
(since all vertices are sure-almost-sure winning in $W_1$), and if that set is nonempty we define $R_k = U_k = W_1''$. 
It follows that the sets $U_1,\dots,U_k$ and $R_1,\dots,R_k$ exist and 
have the properties listed above (in particular $R_i$ being a trap in the appropriate 
subgame). 
\end{proof}
\end{longversion}

Now we analyze the amount of memory in winning strategies for relevant special
cases of parity objectives.

\subsection{No memory for sure coB\"uchi and almost-sure parity}\label{sec:coB-parity}
For the conjunction of sure coB\"uchi and almost-sure parity objectives,
memoryless strategies are sufficient for Player~$1$. To see this,
consider the correctness proof of Algorithm~\ref{algo:sas} in the case
the largest $\Omega_{1}$-priority $d$ is odd (which is the case for the 
coB\"uchi objective). The proof constructs a sure-almost-sure winning
strategy in the winning region, defined over the three subregions
$W_1'$, $B \setminus W_1'$, $W_1''$. 
In $W_1'$, the largest $\Omega_{1}$-priority is~$0$ and the sure parity objective 
is therefore trivially satisfied. 
Hence an almost-sure winning strategy for the second parity objective is good enough,
and thus can be chosen memoryless~\cite{CJH04}. In the region $B \setminus W_1'$, a sure attractor
strategy is used, which is memoryless. Finally, in $W_1''$ we can use an argument
by induction over the number of vertices (memoryless strategies trivially suffice
in empty games) to show that memoryless strategies are sufficient for Player~$1$.

\begin{theorem}
Memoryless strategies are sufficient for combined sure coB\"uchi and almost-sure parity
objectives.
\end{theorem}

\subsection{Finite memory for sure parity and almost-sure B\"uchi}\label{sec:parity-B}
For the conjunction of sure parity and almost-sure B\"uchi objectives,
we show that finite memory is sufficient for Player~$1$,
namely memory of size $O(n\cdot d_1)$ where $n$ is the number of vertices
in the game and $d_1$ is the largest $\Omega_{1}$-priority.
The proof is by induction on $d_1$. The case $d_1 = 0$ reduces to 
just almost-sure parity (the sure parity objective is trivially satisfied),
and the case $d_1 = 1$ corresponds to the conjunction of sure coB\"uchi and 
almost-sure parity (B\"uchi) objectives. In both cases, memoryless
strategies are sufficient, thus memory $O(1)$ (see Section~\ref{sec:coB-parity} 
for the latter). 

We proceed by induction for $d_1 \geq 2$. First, if $d_1$ is even, then we construct a sure-almost-sure
winning strategy as follows. Using the characterisation of Lemma~\ref{lem:char-even},
the winning region $W_1$ must be almost-sure winning for both the parity and B\"uchi objectives,
and it can be decomposed into the sure attractor $A$ to the vertices
with $\Omega_{1}$-priority $d_1$, and the complement $B = W_1 \setminus A$
that induces a sure-almost-sure winning subgame closure $\closure{\Game}_{B}$.
By induction hypothesis, there exists a sure-almost-sure winning strategy
in $\closure{\Game}_{B}$ with memory $O(n\cdot (d_1 - 1))$. 
A sure-almost-sure winning strategy in $W_1$ is defined informally as follows:
as long as the current vertex is in $B$, play according to the sure-almost-sure 
winning strategy in $\closure{\Game}_{B}$. Whenever the current vertex is in $A$,
play the sure attractor strategy until reaching a vertex with $\Omega_{1}$-priority $d_1$,
and then play for $n= \abs{W_1}$ rounds a (memoryless) almost-sure winning strategy
for the B\"uchi objective. The memory of the strategy is $n+O(n\cdot (d_1 - 1))$
thus $O(n\cdot d_1)$\todo{Why can't we do max of the two and restrict the memory to $n$?} 
and we argue that the strategy is sure-almost-sure winning.
To show that the (first) parity objective is satisfied surely, consider an 
arbitrary outcome of the strategy, and either it visits $A$ infinitely often,
and the also $\Omega_{1}$-priority $d_1$ infinitely often, or it eventually
remains in $B = W_1 \setminus A$ and thus follows a the sure-almost-sure 
winning strategy in $\closure{\Game}_{B}$. In both case the outcome satisfies
the (first) parity objective. To show that the B\"uchi objective is satisfied 
almost-surely, notice that, $(1)$ conditional to visiting $A$ infinitely often,
there is a bounded probability at least $\nu^n$, where $\nu$ is the smallest non-zero 
probability in the game,   
to visit a B\"uchi vertex
after each visit to $A$, and thus probability~$1$ to visit a B\"uchi vertex
infinitely often; $(2)$ conditional to visiting $A$ finitely often,
the play eventually remains in $B$ where the B\"uchi objective is satisfied 
with probability~$1$. Hence in both cases, the B\"uchi objective is satisfied 
almost-surely.

Second, if $d_1$ is odd, then the strategy construction in the correctness
proof of Algorithm~\ref{algo:sas} gives a sure-almost-sure winning strategy
with memory $\sum_{i=1}^{k} O(n_i \cdot (d_1-1))$ where  $\sum_{i=1}^{k} n_i = n$,
and thus memory $O(n\cdot d_1)$ is sufficient. 

\begin{theorem}
Finite-memory $O(n\cdot d_1)$ strategies are sufficient for 
combined sure parity and almost-sure B\"uchi objectives, where $n$ is the number of vertices
and $d_1$ is the largest priority in the sure parity objective.
\end{theorem}

\bibliography{biblio} 

@inproceedings{AKV16,
  author       = {S. Almagor and O. Kupferman and Y. Velner},
  title        = {Minimizing Expected Cost Under Hard {B}oolean Constraints, with Applications to Quantitative Synthesis},
  booktitle    = {Proc. of {CONCUR}: Concurrency Theory},
  series       = {LIPIcs},
  volume       = {59},
  pages        = {9:1--9:15},
  publisher    = {Schloss Dagstuhl - Leibniz-Zentrum f{\"{u}}r Informatik},
  year         = {2016}
}

@Book{BK08,
  author    = {C. Baier and J.-P. Katoen},
  title     = {Principles of Model Checking},
  publisher = {MIT},
  year      = {2008}
}

@inproceedings{BRR17,
  author       = {R. Berthon and M. Randour and J.{-}F. Raskin},
  title        = {Threshold Constraints with Guarantees for Parity Objectives in {M}arkov Decision Processes},
  booktitle    = {Proc. of ICALP: International Colloquium on Automata, Languages, and Programming},
  series       = {LIPIcs},
  volume       = {80},
  pages        = {121:1--121:15},
  publisher    = {Schloss Dagstuhl - Leibniz-Zentrum f{\"{u}}r Informatik},
  year         = {2017}
}

@inproceedings{BGR20,
  author       = {R. Berthon and S. Guha and J.{-}F. Raskin},
  title        = {Mixing Probabilistic and non-Probabilistic Objectives in {M}arkov Decision Processes},
  booktitle    = {Proc. of {LICS}: Logic in Computer Science},
  pages        = {195--208},
  publisher    = {{ACM}},
  year         = {2020}
}

@inproceedings{BKW24,
  author       = {R. Berthon and J.{-}P. Katoen and T. Winkler},
  title        = {{M}arkov Decision Processes with Sure Parity and Multiple Reachability Objectives},
  booktitle    = {Proc. of {RP}: Reachability Problems},
  series       = {LNCS 15050},
  pages        = {203--220},
  publisher    = {Springer},
  year         = {2024}
}

@inproceedings{Boker18,
  author       = {U. Boker},
  title        = {Why These Automata Types?},
  booktitle    = {Proc. of {LPAR-22.}: 22nd International Conference on Logic for Programming, Artificial Intelligence and Reasoning, },
  series       = {EPiC Series in Computing},
  volume       = {57},
  pages        = {143--163},
  publisher    = {EasyChair},
  year         = {2018}
}

@article{BFRR17,
  author       = {V. Bruy{\`{e}}re and E. Filiot and M. Randour and J.{-}F. Raskin},
  title        = {Meet your expectations with guarantees: Beyond worst-case synthesis in quantitative games},
  journal      = {Inf. Comput.},
  volume       = {254},
  pages        = {259--295},
  year         = {2017}
}

@inproceedings{CJH03,
  author    = {K. Chatterjee and M. Jurdzinski and T.~A. Henzinger},
  title     = {Simple Stochastic Parity Games},
  booktitle = {Proc. of CSL: Computer Science Logic},
  series    = {LNCS 2803},
  pages     = {100--113},
  publisher = {Springer},
  year      = {2003}
}

@inproceedings{CJH04,
  author    = {K. Chatterjee and M. Jurdzinski and T.~A. Henzinger},
  title     = {Quantitative stochastic parity games},
  booktitle = {Proc. of SODA: Symposium on Discrete Algorithms},
  pages     = {121--130},
  publisher = {{SIAM}},
  year      = {2004}
}

@inproceedings{CD16a,
  author    = {K. Chatterjee and L. Doyen},
  title     = {Perfect-Information Stochastic Games with Generalized Mean-Payoff Objectives},
  booktitle = {Proc. of {LICS}: Logic in Computer Science},
  publisher = {IEEE Computer Society Press},
  year      = {2016},
  pages     = {247-256}
}

@inproceedings{CFKSW13,
  author       = {T. Chen and V. Forejt and M.~Z. Kwiatkowska and A. Simaitis and C. Wiltsche},
  title        = {On Stochastic Games with Multiple Objectives},
  booktitle    = {Proc. of {MFCS}: Mathematical Foundations of Computer Science},
  series       = {LNCS 8087},
  pages        = {266--277},
  publisher    = {Springer},
  year         = {2013}
}

@article{CN06,
  author       = {T. Colcombet and D. Niwinski},
  title        = {On the positional determinacy of edge-labeled games},
  journal      = {Theor. Comput. Sci.},
  volume       = {352},
  number       = {1-3},
  pages        = {190--196},
  year         = {2006}
}

@inproceedings{CJKLS17,
  author       = {C.~S. Calude and S. Jain and B. Khoussainov and W. Li and F. Stephan},
  title        = {Deciding parity games in quasipolynomial time},
  booktitle    = {Proc. of {STOC}: Symposium on Theory of Computing},
  pages        = {252--263},
  publisher    = {{ACM}},
  year         = {2017}
}

@inproceedings{CHP07,
  author       = {K. Chatterjee and T.~A. Henzinger and N. Piterman},
  title        = {Generalized Parity Games},
  booktitle    = {Proc. of {FOSSACS}: Foundations of Software Science and Computational Structures},
  series       = {LNCS 4423},
  pages        = {153--167},
  publisher    = {Springer},
  year         = {2007}
}

@inproceedings{CP19,
  author       = {K. Chatterjee and N. Piterman},
  title        = {Combinations of Qualitative Winning for Stochastic Parity Games},
  booktitle    = {Proc. of {CONCUR}: Concurrency Theory},
  series       = {LIPIcs},
  volume       = {140},
  pages        = {6:1--6:17},
  publisher    = {Schloss Dagstuhl - Leibniz-Zentrum f{\"{u}}r Informatik},
  year         = {2019}
}

@book{Durrett2010,
  title         = {{Probability: Theory and Examples}},
  author        = {R. Durrett},
  year          = 2010,
  publisher     = {Cambridge University Press}
}

@inproceedings{EJ91,
  author       = {E. A. Emerson and C. S. Jutla},
  title        = {Tree Automata, Mu-Calculus and Determinacy (Extended Abstract)},
  booktitle    = {Proc. of FOCS: Foundations of Computer Science},
  pages        = {368--377},
  publisher    = {{IEEE} Computer Society},
  year         = {1991}
}

@book{FV97,
  author    = {J. Filar and K. Vrieze},
  title     = {Competitive {Markov} Decision Processes},
  publisher = {Springer}, 
  year      = {1997}
}

@proceedings{automata,
  editor    = {E. Gr{\"a}del and W. Thomas and T. Wilke},
  title     = {Automata, Logics, and Infinite Games: A Guide to Current Research},
  publisher = {Springer},
  series    = {LNCS 2500},
  year      = {2002}
}

@article{GZ18,
  author       = {M. Guo and M.~M. Zavlanos},
  title        = {Probabilistic Motion Planning Under Temporal Tasks and Soft Constraints},
  journal      = {{IEEE} Trans. Autom. Control.},
  volume       = {63},
  number       = {12},
  pages        = {4051--4066},
  year         = {2018}
}

@inproceedings{JL17,
  author       = {M. Jurdzinski and R. Lazic},
  title        = {Succinct progress measures for solving parity games},
  booktitle    = {Proc. of {LICS}: Symposium on Logic in Computer Science},
  pages        = {1--9},
  publisher    = {{IEEE} Computer Society},
  year         = {2017}
}

@article{Jur98,
  author       = {M. Jurdzinski},
  title        = {Deciding the Winner in Parity Games is in {UP}~$\cap$~co{UP}},
  journal      = {Inf. Process. Lett.},
  volume       = {68},
  number       = {3},
  pages        = {119--124},
  year         = {1998}
}

@inproceedings{Leh18,
  author       = {K. Lehtinen},
  editor       = {A. Dawar and E. Gr{\"{a}}del},
  title        = {A modal {\(\mu\)} perspective on solving parity games in quasi-polynomial time},
  booktitle    = {Proc. of {LICS}: Symposium on Logic in Computer},
  pages        = {639--648},
  publisher    = {{ACM}},
  year         = {2018}
}

@article{Piterman07,
  author       = {N. Piterman},
  title        = {From Nondeterministic {B}{\"{u}}chi and {S}treett Automata to Deterministic Parity Automata},
  journal      = {Log. Methods Comput. Sci.},
  volume       = {3},
  number       = {3},
  year         = {2007}
}

@article{RVWD13,
  author       = {D.~M. Roijers and P. Vamplew and S. Whiteson and R. Dazeley},
  title        = {A Survey of Multi-Objective Sequential Decision-Making},
  journal      = {J. Artif. Intell. Res.},
  volume       = {48},
  pages        = {67--113},
  year         = {2013}
}

@inproceedings{Safra92,
  author       = {S. Safra},
  title        = {Exponential Determinization for omega-Automata with Strong-Fairness Acceptance Condition (Extended Abstract)},
  booktitle    = {Proc. of STOC: Symposium on Theory of Computing},
  pages        = {275--282},
  publisher    = {{ACM}},
  year         = {1992}
}

@InCollection{Thomas97,
  author    = {W. Thomas},
  title     = {Languages, Automata, and Logic},
  booktitle = {Handbook of Formal Languages},
  publisher = {Springer},
  year      = {1997},
  volume    = {3, Beyond Words},
  chapter   = {7},
  pages     = {389-455}
}

@article{VCDHRR15,
  author    = {Y. Velner and K. Chatterjee and L. Doyen and T.~A. Henzinger and A. Rabinovich and J.-F. Raskin},
  title     = {The Complexity of Multi-Mean-Payoff and Multi-Energy Games},
  journal   = {Information and Computation},
  volume    = {241},
  year      = {2015},
  pages     = {177-196},
  publisher = {Elsevier}
}

@article{Zie98,
  author       = {W. Zielonka},
  title        = {Infinite Games on Finitely Coloured Graphs with Applications to Automata on Infinite Trees},
  journal      = {Theor. Comput. Sci.},
  volume       = {200},
  number       = {1-2},
  pages        = {135--183},
  year         = {1998}
}

\end{document}